%% file: paper-different-style.tex
\documentclass[letterpaper,11pt]{article}
\usepackage{amsmath, amssymb}
\usepackage{color}
\usepackage{fullpage}
\usepackage[numbers]{natbib}
\usepackage{tikz}
\usepackage{enumerate}
\usepackage{hyperref}
\usepackage{graphicx}
\usepackage{caption}
\usepackage{subcaption}
\usepackage{zerosum,csquotes}
\usepackage{enumitem,linegoal}
\usepackage{comment}	
\usepackage{ctable}

\usepackage{mathtools}
\usepackage[flushleft]{threeparttable}
\usepackage{algorithm}
\usepackage{algorithmicx}
\usepackage{algpseudocode}

\usepackage{footnote}
\makesavenoteenv{tabular}
\makesavenoteenv{table}

\usepackage[margin=1in]{geometry}

 \newtheorem{theorem}{Theorem}[section]
 \newtheorem{lemma}[theorem]{Lemma}
 \newtheorem{proposition}[theorem]{Proposition}

 \newtheorem{definition}[theorem]{Definition}

\makeatletter
\def\GrabProofArgument[#1]{ #1: \egroup\ignorespaces}
\def\proof{\noindent\textbf\bgroup Proof%
	\@ifnextchar[{\GrabProofArgument}{. \egroup\ignorespaces}}

\makeatother

\input{macros.tex}

\newcommand*\samethanks[1][\value{footnote}]{\footnotemark[#1]}

\newcounter{proccnt}

	\newtheorem{theorem*}{Theorem}
	\newtheorem{definition*}{Definition}
	\newtheorem{corollary*}{Corollary}

	\usepackage{tikz}
	\usepackage{pgfplots}

\title{Price of Competition and Dueling Games}

\author{
	Sina Dehghani \samethanks[1] \samethanks[2]
	\and MohammadTaghi HajiAghayi \thanks{University of Maryland. email: \texttt{\{seddighin,dehghani,mahini\}@umd.edu, hajiagha@cs.umd.edu}}
	\thanks{Supported in part by NSF CAREER award CCF-1053605,  NSF BIGDATA grant IIS-1546108, NSF AF:Medium grant CCF-1161365,   
		DARPA GRAPHS/AFOSR grant FA9550-12-1-0423, another DARPA SIMPLEX grant, and a Google Faculty Research award.}
	\and Hamid Mahini \samethanks[1] \samethanks[2]
	\and Saeed Seddighin \samethanks[1] \samethanks[2]
}

\begin{document}
\renewcommand{\theenumi}{(\roman{enumi}).}
\renewcommand{\labelenumi}{\theenumi}
\sloppy

%\setlength{\abovecaptionskip}{0.1ex}
% \setlength{\belowcaptionskip}{0.1ex}
% \setlength{\floatsep}{0.1ex}
% \setlength{\textfloatsep}{0.1ex}
% 
% 
% \abovedisplayskip.30ex
%   \belowdisplayskip.30ex
%   \abovedisplayshortskip.30ex
%   \belowdisplayshortskip.30ex

\date{}

\maketitle

\begin{abstract}
\input{abstract.tex}
\end{abstract}

\input{intro.tex}
\input{results.tex}

\input{related.tex}

\input{model.tex}
\section{Price of competition in the linear ranking duel} \label{sec:ranking}
\subsection{Welfare maximization ranking duel}\label{sec:maxim}
\input{ranking.tex}

\subsection{Cost minimization ranking duel}
\input{minimization.tex}

\input{general_framework.tex}

\input{ack.tex}

\bibliographystyle{abbrv}

\bibliography{blotto-abrv}
\appendix
\section{A feasible solution for the dual linear program}
\input{lp_table.tex}

\end{document}

%% file: macros.tex
\newcommand{\cpx}[1]{\mathcal{O}(#1)}
\newcommand{\ff}{v}
\newcommand{\ffc}{c}
\newcommand{\triggerv}[2]{v^{#2}_{#1}}
\newcommand{\triggersw}[2]{\SW^{#2}(#1)}
\newcommand{\triggerpoc}[1]{\text{PoC}^{#1}}
\newcommand{\OPT}{\text{OPT}}
\newcommand{\xminmax}{\x^*}
\newcommand{\pin}{\pi^{-1}}
\newcommand{\transformed}[1]{\mathcal{T}(#1)}
\newcommand{\payoff}[1]{u^A_{#1}}
\newcommand{\depth}[2]{d_{#1}(#2)}
\newcommand{\rank}[2]{\pi(#2)}

\newcommand{\Rnonnegative}{\mathbb{R}_{\ge 0}}
\newcommand{\positionutility}[1]{f(#1)}
\newcommand{\SW}{\text{SW}}
\newcommand{\SC}{\text{SC}}

\newcommand{\p}[1]{{p}_{#1}}
\newcommand{\x}{\mathbf {x}}
\newcommand{\px}{s}
\newcommand{\y}{\mathbf {y}}
\newcommand{\f}[2]{{\ff}_{#2}(#1)}
\newcommand{\fc}[2]{{\ffc}_{#2}(#1)}
\newcommand{\uu}[4]{{u}^{#1}_{#2}(#3,#4)}
\newcommand{\fuu}[3]{{u}^{#1}(#2,#3)}
\newcommand{\POC}{\text{$\text{PoC}$}}
\newcommand{\POCBound}{0.612}
\newcommand{\POCBoundnew}{0.637}
\newcommand{\POCUpper}{0.833}
\newcommand{\xb}{\mathbf{x}}

\newcommand{\game}{dueling game}
\newcommand{\piab}{\pi_{ab}}
\newcommand{\piba}{\pi_{ba}}
\newcommand{\vabx}{h_{ab}(\x)}
\newcommand{\vab}{h_{ab}}

\newcommand{\nab}{Pr_{\pi \sim \xb}[\pi(a) < \pi(b)]}
\newcommand{\nba}{Pr_{\pi \sim \xb}[\pi(b) < \pi(a)]}

\newcommand{\POCCost}{\text{PoC}_{cost}}

\newcommand{\rabb}{R^{ab}_b}
\newcommand{\raba}{R^{ab}_a}
\newcommand{\rbab}{R^{ba}_b}
\newcommand{\rbaa}{R^{ba}_a}

\newcommand{\Alpha}{c}
\newcommand{\Beta}{d}

%% file: abstract.tex
We study competition in a general framework introduced by Immorlica,
Kalai, Lucier, Moitra, Postlewaite, and Tennenholtz~\cite{duel11} and answer their main open question. Immorlica et
al.~\cite{duel11} considered classic optimization problems in terms
of competition and introduced a general class of games called
\textit{dueling games}. They model this competition as a zero-sum
game, where two players are competing for a user's satisfaction. In
their main and most natural game, the \textit{ranking duel}, a user
requests a webpage by submitting a query and players output an
ordering over all possible webpages based on the submitted query.
The user tends to choose the ordering which displays her requested
webpage in a higher rank. The goal of both players is to maximize
the probability that her ordering beats that of her opponent and
gets the user's attention. Immorlica et al.~\cite{duel11}  show this
game directs both players to provide suboptimal search results.
However, they leave the following as their main open question:
``does competition between algorithms improve or degrade expected
performance?" (see the introduction for more quotes) In this paper, we resolve this question for the
ranking duel and a more general class of dueling games.

More precisely, we study the quality of orderings in a competition
between two players. This game is a zero-sum game, and thus any Nash
equilibrium of the game can be described by minimax strategies. Let
the value of the user for an ordering be a function of the position
of her requested item in the corresponding ordering, and the social
welfare for an ordering be the expected value of the corresponding
ordering for the user. We propose the {\em price of competition}
which is the ratio of the social welfare for the worst minimax
strategy to the social welfare obtained by a social planner. Finding
the price of competition is another approach to obtain structural
results of Nash equilibria. We use this criterion for analyzing the
quality of orderings in the ranking duel. Although Immorlica et
al.~\cite{duel11} show that the competition leads to suboptimal
strategies, we prove the quality of minimax results is surprisingly
close to that of the optimum solution. In particular, via a novel
factor-revealing LP for computing price of anarchy, we prove if the
value of the user for an ordering is a linear function of its
position, then the price of competition is at least 0.612 and
bounded above by 0.833. Moreover we consider the cost minimization
version of the problem. We prove, the social cost of the worst
minimax strategy is at most 3 times the optimal social cost.

Last but not least, we go beyond linear valuation functions and
capture the main challenge for bounding the price of competition for
any arbitrary valuation function. We present a principle which
states that the lower bound for the price of competition for all 0-1
valuation functions is the same as the lower bound for the price of
competition for all possible valuation functions. It is worth
mentioning that this principle not only works for the ranking duel
but also for all dueling games. This principle says, in any dueling
game, the most challenging part of bounding the price of competition
is finding a lower bound for 0-1 valuation functions. We leverage
this principle to show that the price of competition is at least
$0.25$ for the generalized ranking duel, and to find upper bounds on
the price of competition for the binary search duel and the
compression duel which are introduced by Immorlica et
al.~\cite{duel11}.

%% file: intro.tex
\section{Introduction}
The conventional wisdom is that competition among suppliers will
increase social welfare by providing  consumers with competitive
prices, high-quality products, and a wide range of options. A
classic example is the Bertrand competition~\cite{Bertrand1883}
where suppliers compete in price to incentivize consumers to buy
from them and as a result the market price decreases to the point
that it matches the marginal cost of production. Indeed there are
many theoretical and empirical  studies for supporting this belief
in the economic literature (See, e.g., \cite{ahmadinejad2015forming,micro1,micro2,micro3}).
 However while in
many markets the competition steers businesses to optimize their
solutions for consumers, there are competitive markets in which
businesses do not offer the best option to consumers. An
interesting example for describing this situation is a dueling game,
namely, a zero-sum game where two players compete to attract users.
%A good example for this kind of market is the market for the sponsored search where search engines compete on search result quality.
Immorlica, Kalai, Lucier, Moitra, Postlewaite, and  Tennenholtz
\cite{duel11} showed surprisingly if players are aimed
to beat their opponents in a dueling game, they may offer users
suboptimal results.
%
%For the sake of simplicity, we focus our attention to the case with two competitors in the market.
However, they raised this question regarding the efficiency of the
competition as the authors write, {\em ``Perhaps more importantly,
one could ask about performance loss inherent when players choose
their algorithms competitively instead of using the (single-player)
optimal algorithm. In other words, what is the price of anarchy \footnote{Indeed Immorlica et al. \cite{duel11} use
	the term of the {\em price of anarchy} in their aforementioned open
	question for the same concept of the {\em price of competition} in
	this paper.}  of
a given duel? ... Our main open question is (open question 1): does
competition between algorithms improve or degrade expected
performance?''} As we describe below, we study this  open question
for a set of dueling games and in particular for the ranking duel
which is an appropriate representative of dueling games due to
Immorlica et al. \cite{duel11}.

{\bf Dueling games.}
%We go beyond the sponsored search competition and investigate the efficiency of competition in dueling games introduced in \cite{duel11}\footnote{Note that the sponsored search problem is a special case of dueling games.}.
%
 A {\em \game\ $\mathcal{G}$} is a
zero-sum game where two players compete for the attention of a
user~\footnote{One can see the user as a population of users with
the same behavior.}. In a \game\  both players try to beat the other
player and offer a better option with a higher value to the user. In
particular, while the user's request is unknown to both players and
they only have access to probability distribution $p$, the goal for
each player is to maximize the probability  that her offer is better
than her opponent's offer. This framework falls within a general and natural class of ranking
or social context games \cite{AKT08,BFHS09}, where each player plays a base game separately and then ultimate
payoffs are determined by both their own outcomes and the outcomes of others. Immorlica et al. argue that this
class of games models a variety of scenarios of competitions between algorithm designers, such as,
competition between search engines (who must choose how to rank search results), or competition between
hiring managers (who must choose from a pool of candidates in the style of the secretary problem).

To be more precise a dueling game is
defined by $4$-tuple $\mathcal{G}=(\Omega, p, S, \ff)$, where
$\Omega$ is the set of all possible requests from the user, $p$ is a
probability distribution over set $\Omega$ i.e., $\p{\omega}$ is the
probability of requesting $\omega  \in \Omega$ by the user, $S$ is
the set of all possible pure strategies for both players, and
$\f{\px}{\omega}$ is the value of
pure strategy $\px \in S$ for the user upon request $\omega \in \Omega$. Note that $v$ is usually considered to be the valuation of the players, but in this paper valuation function $v$ denotes the value for the user. %
While a mixed strategy is a probability distribution over all possible pure strategies in $S$,
we write the value of mixed strategy $\x$ as $\f{\x}{\omega} = E_{\px \sim \x}[\f{\px}{\omega}]$.

A social planner is often interested in choosing a strategy which
maximizes the social welfare, even though it may be a bad strategy
in the competition between players. This means the social welfare
maximizer strategy may not appear in any Nash equilibrium of the
game, and thus the competition between players results in a
suboptimal outcome for the users.
%In this context, a key challenge for a social planner is to measure the efficiency of competition by analyzing the social welfare for any competitive strategy.
%
Knowing the fact that Nash equilibria of a dueling game can be
formed by  suboptimal strategies, the following question seems to be
an important question to ask regarding the inefficiency of this
competition:
%Immorlica et al. \cite{duel11} mentioned the following question regarding the efficiency of Nash equilibria for further studies:
\begin{quote}
What is the social welfare of any Nash equilibrium in a dueling game in comparison to the social welfare of the optimal strategy?
\end{quote}

{\bf Price of competition.}
As aforementioned while in so many cases the competition motivates
businesses to optimize their solutions for consumers, there are
competitive markets and in particular dueling games of Immorlica et
al.~\cite{duel11}, in which businesses do not offer the best option
to consumers. We define \textit{price of competition} in this paper to
capture this phenomenon.

%%It is a general belief
%The conventional wisdom is that competition among suppliers will
%increase social welfare by providing  consumers with competitive
%prices, high-quality products, and a wide range of options.
%%
%There are many theoretical and empirical  studies for supporting
%this belief in the economic literature (See, e.g.,
%\cite{micro1,micro2,micro3}).
%%
%A classic example for describing this  phenomenon is the Bertrand
%competition where suppliers of a homogeneous good compete on price
%and try to incentivize  consumers which buy from the supplier with
%the lowest price \cite{Bertrand1883}.  This competition will
%decrease the market price to the point that it matches the marginal
%cost of production. \xxx{Hamid: I'd suggest to remove this paragraph
%or put it as the first paragraph of the intro.}
%%

First we note that since dueling games are two-player zero-sum
games, Nash equilibria of these games are  characterized by minimax
strategies. Therefore, one can measure the inefficiency of any Nash
equilibrium by comparing the welfare of any minimax strategy, in a
game of competition between two players, with the welfare achieved
by a social welfare maximizer.
%
%In order to quantify the efficiency of the competition, we propose a new concept which is defined as follows. Let $f(i)$ be the value of the user for being ranked at position $i$, and the social welfare of any strategy be the expected value of this strategy for the user.
We are now ready to define the following criterion for measuring the
quality of minimax strategies in a dueling game.
\begin{definition}
{\bf Price of competition (\POC)} is the ratio between the social welfare of the worst
minimax strategy and the social welfare of the best possible
strategy.
\end{definition}

The proposed concept of the price of competition  has the same
spirit as the concept of the price of anarchy, and both concepts try
to measure the inefficiency of Nash equilibria quantitatively.
The price of anarchy,  introduced by the seminal work of
Koutsoupias and Papadimitriou \cite{kp99}, is a well-known concept in game theory that measures
the ratio of the social welfare of the worst Nash
equilibrium to the optimal social welfare.
Although these two concepts  are defined to capture properties of
Nash equilibria, they are meaningfully different.
%\begin{itemize}
%\item
%The price of anarchy measures the social welfare of the
%worst equilibrium ``outcome''. However, the price of competition
%measures the social welfare of the worst minimax ``strategy''.
%\item
In the price of anarchy, the social welfare is defined as
the expected utility of all players in an equilibrium ``outcome''\footnote{Which is essentially
the same as the sum of utilities of all players.} which is always
zero for any zero-sum game. However, in the price of competition,
the social welfare is the expected utility of the user (which is not
a player) in a minimax ``strategy''. In fact, the price of competition is aimed to analyze the
impact of the competition between players on an external user.
%\end{itemize}
%

Since the price of  competition captures the inefficiency of minimax
strategies in two-player zero-sum games and all Nash equilibria of
any two-player zero-sum  game can be described by the set of minimax
strategies, we believe the price of competition  sheds new light on
the structural analysis of Nash equilibria in two-player zero-sum
games.
Indeed as Alon, Demaine, Hajiaghayi, and Leighton \cite{creation4}
mention  understanding the structure of Nash equilibiria, and not
just the price of anarchy,  is very important in general and thus
our work is exactly toward this direction.

Due to the space constraints, all the missing proofs are provided in the appendices.
%As Alon, Demaine, Hajiaghayi, and Leighton \cite{creation4} believe the diameter is an important criterion beside the price of anarchy for understanding the structure of Nash equilibria in network creation games \cite{creation1}, we think the price of competition is an important criterion for analyzing the inefficiency of Nash equilibria in two-player zero-sum competitions.

%\xxx{Hamid to Sina: I have removed the definition of the ranking duel from the intro. I think you should add it to the our results section.}\\
%\xxx{Hamid to Sina: I'd suggest to substitute "the ranking duel" for "the sponsored search game" in the abstract, our results, related work sections.}

%% file: results.tex
\subsection{Our results}
{\bf Ranking duel:} 
To define the ranking duel more precisely, consider a ranking duel with two players. When a user submits a query to a player, she is basically searching a webpage which is unknown to the player. The player only has a prior knowledge about the requested webpage, i.e., for each webpage the probability that this webpage is requested by the user is known. The strategy of
each player is an ordering for displaying webpages. When the requested webpage is realized, the player which puts this webpage in a higher rank gets the user attention, and thus wins the competition. The goal of each player is to maximize the probability of winning the competition. In this situation, a social planner who wants to minimize the expected rank of the requested webpage lists
webpages in a decreasing order of their probabilities. However, this strategy may lose the competition to
another strategy. \footnote{For example consider a situation when the user submits a query and she is interested in
webpages $w_1$, $w_2$, and $w_3$ with probabilities 0.35, 0.33 and 0.32 respectively. In this situation the social
planner ranks webpage $w_i$ at position $i$, for $i = 1, 2, 3$. However, if a player plays based on this strategy, her
opponent puts webpages $w_2$, $w_3$, and $w_1$ at positions 1, 2, and 3 respectively, and thus wins the competition
when the user requests webpages $w_2$ or $w_3$. This means the social planner strategy loses the competition
with probability 0.65.}

%We first study the quality of orderings in the ranking duel.
%Recall the game between two players, where each player tries to beat her opponent, is a zero-sum game, and thus both players would essentially play a minimax strategy in an equilibrium outcome. 
%Regarding this fact, 
We first investigate the quality of minimax strategies and prove that surprisingly the social welfare of any minimax strategy is not far from that of the optimal solution; it is $\POCBound$ of the optimal solution for the linear valuation functions and $0.25$ of the optimal for any arbitrary valuation function.

\begin{theorem*}
\label{thm-rank}
%The social welfare of any minimax strategy of the sponsored search game is at least $0.6$ that of any possible strategy.
Consider an instance of the raking duel. If the valuation function is a non-negative linear function of the rank,
% i.e., the value of the user for being ranked at position $i$ is $f(i) = \Alpha (n-i) + \Beta$ for any $\Alpha, \Beta \geq 0,$
the price of competition is  at least $\POCBound$ for $|\Omega| \geq 10$, and at most $\POCUpper$.
%$\POC$ of the linear raking duel game is at least $\POCBound$ for $|\Omega| \geq 0$.
\end{theorem*}
Our proof needs a careful understanding of properties for minimax strategies and has three main steps. First, we prove nice structural properties of minimax strategies. This step is the main step toward bounding the price of competition and gives an insight into properties of the polytope of minimax strategies.
For example for every two webpages $\omega_1$ and $\omega_2$ with $p_{\omega_1} > p_{\omega_2}$, we prove there is a lower bound on the probability that any minimax strategy ranks webpage $\omega_1$ before webpage $\omega_2$. In the next step, we leverage these properties to write a factor-revealing mathematical program for bounding \POC. 
At last, we find a linear program where the set of its feasible solutions is a superset of the set of feasible solutions of the former mathematical program. We find the optimal solution of this linear program to formally prove the theorem for $|\Omega| \geq 10$. Moreover, we write a computer program to find the optimal solution of the corresponding linear program and show the price of competition is at least $\POCBoundnew$ for $|\Omega| \geq 100$ (which is slightly better the case that $\Omega \geq 10$).
To the best of our knowledge, we are the first to use factor-revealing techniques to bound the inefficiency of equilibria.%compute the price of anarchy.
%Therefore, theorem 1 capture main properties of any minmax strategy.

Afterwards, we consider the cost minimization version of the ranking duel and prove a constant upper bound for the social cost of the game using the same technique of Theorem \ref{thm-rank}. Note that the only difference between the cost minimization and welfare maximization of a dueling game is that function $v$ is a cost function rather than a valuation function, and once a webpage is searched the winner of the cost minimization game is the player who provides a solution with a lower cost. 
Moreover, we define the $\POCCost$ of the ranking duel as the ratio between the minimax strategy with the highest cost and the strategy with the least cost.
In the following theorem we show that $\POCCost \leq 3$. 
\begin{theorem*}
	\label{costmin2}
	For a ranking duel with a linear cost function, we have $\POCCost \leq 3$.
\end{theorem*}
%Note that the polytope of minimax strategies remains unchanged in the cost minimization version in comparison to the welfare maximization version. 
It is worth mentioning that the structural properties of minimax strategies do not depend on the valuation function, and thus the polytope of minimax strategies remains unchanged for every valuation function 
which is a decreasing (increasing) function of rank in the welfare maximization (cost minimization) variant of the game.
%which is a decreasing function of rank if we consider the welfare maximization version or an increasing function of rank in the cost minimization variant of the game. 
Therefore, we leverage structural properties of the polytope of minimax strategies, which is presented in Theorem \ref{thm-rank}, for proving Theorem \ref{costmin2} and in general one can apply our techniques for characterizing the polytope of minimax strategies for an arbitrary valuation function.
Nevertheless, writing the factor-revealing mathematical program totally depends on the linearity of the valuation function.

{\bf General valuation functions:} 
There are situations where the value of the user is not a linear function of rank. For example, consider a user that only cares about the top search results and will be satisfied if and only if her requested webpage is ranked higher than a certain threshold.
We investigate the efficiency of minimax strategies for any non-negative non-linear valuation function. Moreover, we go beyond the ranking duel and consider other dueling games, in the pioneering work of \cite{duel11}.
%, which is a generalization of the sponsored search game.
%
While bounding the social welfare for arbitrary valuation functions and general dueling games seems to be challenging, we present a general principle to capture the main challenge of this problem. The proposed principle has the same spirit as the classic 0-1 principle in the sorting network which states: ``a sorting network will sort any given input if and only if it sorts any given 0-1 input \cite{CLRS}.'' The following principle has the same message and shows if one can bound the social welfare for any 0-1 valuation function, the same bound holds for any arbitrary valuation function. This means the main challenge for bounding the social welfare is to bound it for 0-1 valuation functions. The main idea for proving Theorem \ref{thm-duel} is to decompose any valuation function into 0-1 valuation functions.
\begin{theorem*}
\label{thm-duel}
0-1 Principle: Consider a dueling game. If the price of competition is greater than $\alpha$ when the social welfare is defined based on any 0-1 valuation function, then it is greater than $\alpha$ when the social welfare is defined based on any valuation function.
\end{theorem*}

One can leverage this principle to analyze the efficiency of competition in any dueling game. For example, we show that the price of competition in the ranking duel is at least $0.25$ for an arbitrary valuation function.
\begin{theorem*}
	\label{thm-rank4}
	The price of competition is at least $0.25$ for the ranking duel, when the social welfare is defined based on an arbitrary valuation function.
\end{theorem*}

In the proof of Theorem \ref{thm-rank4}, based on the 0-1 principle, we first consider the problem with pseudo-valuation functions in which the value of each position is either 0 or 1. We consider $\xminmax$ as the minimax strategy with the least social welfare and construct a response strategy $\x'_i$ for the second player for every $1 \leq i \leq n$ as follows:
\begin{itemize}
	\item[] Draw a permutation randomly based on strategy $\xminmax$. If the value of the position of the $i$-th webpage is 1 then play that permutation. Otherwise, swap the position of the $i$-th webpage with one of the positions with value 1 at random and play the new permutation.
\end{itemize}
Next, we use the fact that minimax strategy $\xminmax$ does not lose to strategy $\x'_i$ for proving a set of inequalities which later on helps us to bound the price of competition. Finally, we use the 0-1 principle to show that this lower bound holds for all possible valuation functions.

This principle also helps us to provide upper bounds on the price of competition when one considers a general valuation function. 
For example we show that the $\POC$ of the following two games introduced by Immorlica et al. \cite{duel11} cannot be bounded by any constant value:
\begin{itemize}
\item{\em Binary search duel:} The binary search duel is a dueling game where each player chooses a binary search tree over the set of all possible requests $\Omega$. When the user's request $\omega \in \Omega$ is realized, the value for each strategy is defined based on the depth of request $\omega$ in the corresponding binary search tree.

\item{\em Compression duel:} The compression duel is a dueling game where each player chooses a binary tree over the set of all possible requests $\Omega$, i.e., the set of all leaves of the binary tree would be equal to the set of all possible requests. When the user's request $\omega \in \Omega$ is realized, the value for each strategy is defined based on the depth of request $\omega$ in the corresponding binary tree.
\end{itemize}

\begin{theorem*}
The price of competition is $\cpx{\frac{1}{|\Omega|}}$ for the binary search duel and unbounded for the compression duel, when the social welfare is defined based on an arbitrary valuation function.
\end{theorem*}
In order to construct bad instances for these duels, we design a valuation function which is $1$ for low depths, $0$ for high depths, and a small positive value $\epsilon$ in between. 
We show the price of competition is less than any given number $\beta>0$ for the binary search duel  by constructing an instance of the binary search duel with $|\Omega| = \Theta(\frac{1}{\beta})$.
However, we present an instance of the compression duel with a constant size set $\Omega$ and a price of competition less than any given value $\beta > 0$.

%% file: related.tex
\subsection{Related work}
Immorlica et al.~\cite{duel11} are the first who considered the concept of
dueling games. They present dueling games in the context of dueling
algorithms, where two competitive algorithms try to maximize the
probability of outperforming their opponent for an unknown
stochastic input.
%They present  the ranking duel as an example of dueling games which is exactly the same as our sponsored search game.
While we employ the same model in this paper, our goal completely
differs from that of Immorlica et al. \cite{duel11}. Immorlica et al. \cite{duel11} present
polynomial-time algorithm for finding a minimax strategy of a
dueling game when the polytope of minimax strategies can be
represented by a polynomial number of linear constraints. Knowing
the fact that the polytope of minimax strategies of any ranking duel
has polynomially many facets, they propose a polynomial-time
algorithm for finding a minimax strategy of ranking duels.  This method was later generalized by \cite{ahmadinejad2016duels} to solve the Colonel Blotto game.
Immorlica et al. \cite{duel11} leave the problem of analyzing the social welfare of
competitive algorithms as their main open question. In this paper,
we do not deal with the computational complexity of finding minimax
strategies, but we focus on answering the posted open question and
analyze the social welfare of minimax strategies for a given duel.

As we are interested in quantifying the inefficiency of Nash
equilibria, our proposed concept of the price of competition has the
same flavor as the concept of the price of anarchy
\cite{kp99,AGTBook}. The price of anarchy is commonly used for
quantifying the inefficiency of a system which is constructed by
selfish agents. For example, it has been used to analyze the
inefficiency of Nash equilibria in congestion games
\cite{rout1,rout2}, network creation games
\cite{creation1,creation2,creation3,creation4}, and selfish
scheduling games \cite{sch1,sch2}. (See, e.g. \cite{AGTBook} for
more examples).

Kempe and Lucier \cite{www14} recently study the impact of competition on the social
welfare in a competitive sponsored search market. In their model,
which is a departure from the model of Immorlica et al. \cite{duel11}, search
engines again compete to obtain more users. A user's request is
defined by a set $S$ of webpages which is unknown to search engines,
and the user is satisfied if and only if at least one of webpages in
$S$ is ranked in a better position than a given threshold $t$. The
strategy of each search engine is an ordering over all possible
webpages. At last, the user chooses a search engine based on a
selection rule which is a function of probability of being satisfied
by each search engine. Kempe and Lucier \cite{www14} prove that if search engines
extract utility from satisfied users or the search engine selection
rule is convex, then the social welfare of the game is at least half
of the optimum social welfare. Moreover, they show if the utility of
search engines is driven from all customers and the search engine
selection rule is concave, then the social welfare of the
game is bounded away from that of the optimum solution by a factor
of $\Omega(n)$, where $n$ is the number of all possible webpages. We
would like to note that our model is a general model for studying all dueling games
%in the ranking duel, 
which is exactly the
same as the model of Immorlica et al. \cite{duel11}, and is significantly different from
that of Kempe and Lucier \cite{www14}.

There is a line of research that study a competition between
advertisers in sponsored search auctions \cite{ad1,ad4,ad5,ad2,ad3}.
These works analyze the revenue of a single search engine in various
settings regarding users' behavior and the business model of
advertisers. However, in ranking duel we investigate a competition
between players who provide orderings rather than advertisers.

There is a rich literature in economics that explains product
differentiation in competitive markets. While producing similar
products is supported by classical models such as the Hotelling
model \cite{hotelling}, Aspremont, Gabszewicz, and Thisse
\cite{AGT79} argue that competitive producers may improve their
revenue by producing different products.  See, e.g.,
\cite{micro1,micro2,micro3} for details on this literature. The same
phenomenon can be seen in the sponsored search market, e.g., Telang,
Rajan, and Mukhopadhyay \cite{ss04} show low-quality search engines
may extract revenue from the sponsored search market.

%% file: model.tex
\section{Model}
\subsection{Dueling games}\label{sec:2.1}
%A {\em \game\ } is a zero-sum two-player game and is defined by $4$-tuple $\mathcal{G}=(\Omega, p, S, \ff)$, where $\Omega$ is the set of all possible events, $p$ is a probability distribution over set $\Omega$, i.e., $\p{\omega}$ is the probability of event $\omega  \in \Omega$, $S$ is the set of all possible pure strategies of both players, and $\f{\px}{\omega}$ is
%the value of pure strategy $\px \in S$ when event $\omega \in \Omega$ has occurred. A mixed strategy is a probability distribution over all possible pure strategies $S$.

In \game\ $\mathcal{G}$ both players  try to beat the other player and offer a better value in the competition. Assume players $A$ and $B$ play pure strategies $\px_A$ and $\px_B$ respectively, and event $\omega$ has occurred. In this situation, player $A$ wins the competition if and only if $\f{\px_A}{\omega} > \f{\px_B}{\omega}$, and thus the utility of player $A$ given event $\omega$ can be written as follows:
\begin{equation}\nonumber
\uu{A}{\omega}{\px_A}{\px_B}=
\begin{cases}
+1 & \text{if } \f{\px_A}{\omega} > \f{\px_B}{\omega} \\
0 & \text{if } \f{\px_A}{\omega} = \f{\px_B}{\omega}\\
-1 & \text{if } \f{\px_A}{\omega} < \f{\px_B}{\omega}\\
\end{cases}
\end{equation}
Now consider a situation where players $A$ and $B$ play mixed strategies $\x$ and $\y$ respectively and event $\omega$ has occurred. The utility of player $A$ is the probability that player $A$ wins the competition minus the probability that player $B$ wins the competition and can be defined as follows:
\[
\uu{A}{\omega}{\x}{\y}= Pr_{\begin{subarray}{l}\px_A \sim \x \\ \px_B \sim \y\end{subarray}}  [\f{\px_A}{\omega} > \f{\px_B}{\omega}] - Pr_{\begin{subarray}{l}\px_A \sim \x \\ \px_B \sim \y\end{subarray}}[\f{\px_A}{\omega} < \f{\px_B}{\omega}]
\]

Finally the overall utility of player $A$ is $\fuu{A}{\x}{\y}=\sum_{\omega}{\p{\omega}\uu{A}{\omega}{\x}{\y}}$.
Since dueling game $\mathcal{G}$ is a zero-sum game the utility of player $B$ is the negation of the utility of player $A$ for each $\omega$, i.e., $\uu{B}{\omega}{\x}{\y} = - \uu{A}{\omega}{\x}{\y}$ and thus $\fuu{B}{\x}{\y} = - \fuu{A}{\x}{\y}$.

\begin{definition}
{\bf Minimax strategy:} Strategy $\x$ of player $A$ is minimax if $\x \in \text{argmax}_{\x'} \{ \text{min}_{\y}\{\fuu{A}{\x'}{\y} \}\}$. Similarly, Strategy $\y$ of player $B$ is minimax if 
$\y \in \text{argmax}_{\y'} \{ \text{min}_{\x}\{\fuu{B}{\x}{\y'} \}\}$.
\end{definition}
Based on the definition of dueling games and the fact that the set of all possible pure strategies for both players is $S$, we can conclude the outcome of both players in any Nash equilibrium is $0$ and moreover the set of minimax strategies of both players coincide. We define the set of minimax strategies by $\mathcal{M}$.
%Moreover, duel game $\mathcal{G}$ is a zero-sum game. This means pair $(\x, \y)$ is a Nash equilibrium of the game if and only if both $\x$ and $\y$ are minimax strategies. 

%In other words, the value of \game\ $\mathcal{G}$ depends on an element of uncertainty which is defined by probability distribution $p$ over set of events $\Omega$.
\begin{definition}
{\bf Social welfare:} Consider \game\ $\mathcal{G}=(\Omega, p, S, \ff)$. The social welfare of pure strategy $\px$ is the expected value of this strategy over all possible events and can be written as $\SW(\px) = \sum_{\omega} \p{\omega}\f{\px}{\omega}$. The social welfare of mixed strategy $\x$ is $\SW(\x) = E_{\px \sim \x}[\SW(\px)]$.
\end{definition}

\iffalse
Now consider the following scenario. 
\begin{itemize}
\item A customer arrives and her request is defines based on event $\omega$ in $\Omega$. Event $\omega$ is unknown to both players, and they only has prior knowledge $p$.

\item Both players select their strategies and compete in a dueling game as defined above. 

\item The customer looks at both offers and selects the best offer. The customer will lock into the player which offer her the best and select this player for her further requests.
\end{itemize}
\fi
In this paper, we are interested to study the social welfare of the game in equilibria. 
Note that the customer locks into one of the players in long term. On the other hand, both players only try to offer the customer a better option than the other one, and thus play a minimax strategy in the competition.  These cause inefficiency in the game. Here we define a new criterion to measure this inefficiency in the game.
\begin{definition}
{\bf Price of competition:} The price of competition is the ratio of the worst minimax strategy to the optimal solution which is: 
\[\frac{\min_{\x \in \mathcal{M}} \SW(\x)}{\max_{\x} \SW(\x)} = \frac{\min_{\x \in \mathcal{M}} \SW(\x)}{\max_{\px \in S} \SW(\px)}.\]
\end{definition}

Similar to the welfare maximization model, we consider the cost minimization model in which players try to beat the opponent by offering a lower cost to the user. In particular we have a cost function $c$, such that $c_{\omega}(s)$ denotes the cost of strategy $s$ and event $\omega$. Hence, the utility of player $A$ would be defined as
\[
\uu{A}{\omega}{\x}{\y}= Pr_{\begin{subarray}{l}\px_A \sim \x \\ \px_B \sim \y\end{subarray}}  [\fc{\px_A}{\omega} < \fc{\px_B}{\omega}] - Pr_{\begin{subarray}{l}\px_A \sim \x \\ \px_B \sim \y\end{subarray}}[\fc{\px_A}{\omega} > \fc{\px_B}{\omega}].
\]
Similarly we define the social cost $\SC(s)=\sum_{\omega}p_{\omega}c_{\omega}(s)$ for a pure strategy $s$ and $\SC(\x) = E_{\px \sim \x}[\SC(\px)]$ for a mixed strategy $\x$. Finally the price of competition in cost minimization version is defined as 
\[\frac{\max_{\x \in \mathcal{M}} \SC(\x)}{\min_{\x} \SC(\x)} = \frac{\max_{\x \in \mathcal{M}} \SC(\x)}{\min_{\px \in S} \SC(\px)}.\]
\subsection{Ranking duel}
Ranking duel is a \game\ where $\Omega=\{1, \cdots, n\}$ is the set of $n$ webpages which can be requested by a user. In this game,  the set of pure strategies $S$ is equal to the set of all possible permutations over $\Omega$, i.e., each player outputs an ordering of webpages for the user. We denote each pure strategy  of the ranking duel by $\pi$ (instead of $\px$) where $\pi(\omega)$ is the rank of webpage $\omega$. The valuation function $v$ of a raking duel can be defined based on function $f:\{1,\cdots,n\} \rightarrow R^+\cup\{0\}$ as $\f{\pi}{\omega} = f({\pi}(\omega))$.
Consider mixed strategy $\x$ where $\x_{\pi}$ is the probability that strategy $\x$ outputs permutation $\pi$. The {\em social welfare} of strategy $\mathbf{x}$ can be defined as:
\begin{equation}
\SW(\mathbf{x}) = \sum_{\omega}\sum_{\pi} \p{\omega} {\x}_{\pi} f({\pi}(\omega)).
\end{equation}
%\iffalse
%\begin{definition}
%{\bf Optimum strategy:} Consider an instance of the ranking duel with $n$ webpages with probabilities $p_1\geq p_2 \geq \cdots\geq p_n$. For non-increasing function $f$, the optimum strategy which maximizes the social welfare puts webpage $i$ at position $i$, and thus its social welfare is $\sum_{i}p_i f(i)$.
%\end{definition}
%\fi

%\iffalse
%\begin{definition}
%Let $\mathbf{x}$ be a strategy of the ranking duel. $y(\mathbf{x})_{b \succ a} = y(\mathbf{x})_{a \prec b}$ denote the probability that for two drawn permutations $\pi$ and $\pi'$ from $\mathbf{x}$, $\pi_a > \pi'_b$. Moreover, $y(\mathbf{x})_{b \simeq a} $ denotes the probability that rank of webpages $a$ and $b$ be the same in two drawn permutations $\pi$ and $\pi'$ from $\mathbf{x}$.
%\end{definition}
%\fi

%% file: ranking.tex
In this section we give bounds for the $\POC$ in the ranking duel when the valuation function is non-negative and linear, in other words $f(i)=\Alpha(n- i)+\Beta$, where $\Alpha, \Beta \geq 0$. 

First we formulate the social welfare of strategy $\x$ and the optimal social welfare. Without loss of generality in this section we assume $p_1 \geq p_2 \geq \ldots \geq p_n$. Let $Pr_{\pi \sim \xb} [\pi(a) = i]$ denote the probability that in a randomly drawn permutation $\pi$ from strategy $\x$, the rank of webpage $a$ is $i$. Similarly let $Pr_{\pi \sim \xb} [\pi(a) < \pi(b)]$ denote the probability that in a randomly drawn permutation $\pi$ from strategy $\x$, webpage $a$ comes before webpage $b$.
\begin{proposition} \label{prp-1}
In a ranking duel with valuation function $f$ and $n$ webpages, the social welfare of a strategy $\x$ is $\SW_f(\x)=\sum_{a=1}^n\sum_{i=1}^np_a{Pr_{\pi \sim \xb} [\pi(a) = i]f(i)}.$
\end{proposition}
\begin{proof}
	For any page $a$ and position $i$ we compute the probability that $a$ is chosen and it is located at position $i$ times $f(i)$. Hence,
	\begin{equation*}
	\SW(\x) = \sum_{a=1}^n\sum_{i=1}^n p_a Pr_{\pi \sim \xb}[\pi(a) = i] f(i)
	\end{equation*}
\end{proof}

Let $\OPT$ be the strategy with the maximum social welfare. Hence $\SW(\OPT)$ is formulated as follows.
\begin{proposition}\label{prp-2}
In a ranking duel with valuation function $f$ and $n$ webpages, the optimal social welfare is $\SW_f(\OPT)= \sum_{a=1}^n{p_af(a)}.$
\end{proposition}
\begin{proof}
	The optimal strategy is to sort the pages by descending order of their probability, therefore $\SW_f(\OPT)$ is equal to the social welfare of permutation $\pi=\langle 1, 2, \ldots, n \rangle$. Thus,
	\begin{align*}
	\SW_f(\OPT)=\sum_{a=1}^n{p_af(a)}.
	\end{align*}
\end{proof}

Lemma \ref{linearity} shows that for any minimax strategy $\x$ and any linear function $f(i)=\Alpha(n- i) + \Beta$ with $\Alpha, \Beta \geq 0$, the $\POC$ is no less than the case in which $f(i)=n-i$.

\begin{lemma}\label{linearity}
For valuation functions $f(i)=n-i$, $f'(i)=\Alpha(n- i) + \Beta$ with $\Alpha, \Beta \geq 0$, and any strategy $\x$,  $\frac{\SW_f(\x)}{\SW_f(\OPT)} \leq \frac{\SW_{f'}(\x)}{\SW_{f'}(\OPT)}.$
\end{lemma}
\begin{proof}
	By Propositions \ref{prp-1} and \ref{prp-2}, we have
	\begin{align*}
	\frac{\SW_f(\x)}{\SW_f(\OPT)} &=\frac{\sum_{a=1}^n\sum_{i=1}^n p_a Pr_{\pi \sim \xb}[\pi(a) = i] (n-i)}{\sum_{a=1}^n{p_a(n-a)}} & \text{By multiplying the sides by $\Alpha$}\\
	&=\frac{\sum_{a=1}^n\sum_{i=1}^n p_a Pr_{\pi \sim \xb}[\pi(a) = i] \Alpha(n-i)}{\sum_{a=1}^n{p_a\Alpha(n-a)}} & \text{Since the fraction is less than $1$}\\
	&\leq \frac{\Beta+\sum_{a=1}^n\sum_{i=1}^n p_a Pr_{\pi \sim \xb}[\pi(a) = i] \Alpha(n-i)}{\Beta+\sum_{a=1}^n{p_a\Alpha(n-a)}} & \text{Since $\sum_{i=1}^n p_a Pr_{\pi \sim \xb}[\pi(a) = i] = \sum_{a=1}^n{p_a}=1$}\\
	&= \frac{\sum_{a=1}^n\sum_{i=1}^n p_a Pr_{\pi \sim \xb}[\pi(a) = i] (\Alpha(n-i) + \Beta)}{\sum_{a=1}^n{p_a(\Alpha(n-a)+\Beta)}} \\%& \text{Since $\sum_{i=1}^n p_a Pr_{\pi \sim \xb}[\pi(a) = i] = \sum_{a=1}^n{p_a}=1$}\\
	%									&= \frac{\sum_{a=1}^n\sum_{i=1}^n p_a Pr_{\pi \sim \xb}[\pi(a) = i] (\Alpha(n-i) + \Beta)}{\sum_{a=1}^n{p_a(\Alpha(n-a)+\Beta)}}\\
	&=\frac{\SW_{f'}(\x)}{\SW_{f'}(\OPT)}.									
	\end{align*}
\end{proof}

 Thus any lower bound for the $\POC$ with $f(i)=n-i$, is also a lower bound for the $\POC$ with any other linear valuation function. Therefore, from now on we assume $f(i)=n-i$, and use $\SW(\x)$ and $\SW(\OPT)$ instead of $\SW_f(\x)$ and $\SW_f(\OPT)$, respectively. Hence $\SW(\OPT)=\sum_{a=1}^n{p_a(n-a)}.$ Now we try to compute $\SW(\x)$ from a different perspective.
\begin{proposition} \label{prp:1}
In a ranking duel with $n$ webpages, the social welfare of strategy $\x$ is $$\SW(\x)=\sum_{a=1}^n \sum_{b=a+1}^n p_a  Pr_{\pi \sim \xb}[\pi(a) <\pi(b)] + p_bPr_{\pi  \sim \xb}  [\pi(b)  <\pi(a)] .$$
\end{proposition}
\begin{proof}
	%First for any page $a$ and position $k$ we compute the probability that $a$ is chosen and it is located at position $k$ times $f(k)$. Hence,
	By Proposition \ref{prp-1},
	\begin{align*}
	\SW(\mathbf{x}) &= \sum_{a=1}^n\sum_{i=1}^n p_a Pr_{\pi \sim \xb}[\pi(a) = i] (n-i) \\
	&= \sum_{a=1}^n p_a \sum_{i=1}^n Pr_{\pi \sim \xb}[\pi(a) = i] (n-i) \quad\quad \quad\quad\quad\quad \quad\quad\ \ \text{for each page $a$ at position $i$ in $\pi$,}\\ 
	&\quad\quad\quad\quad \quad\quad\quad\quad \quad\quad\quad\quad \quad\quad\quad \quad\quad\quad\quad \quad\quad \text{consider $(n-i)$ pages at higher positions}\\
	&= \sum_{a=1}^n p_a \sum_{b=1}^n Pr_{\pi \sim \xb}[\pi(a) <\pi(b)] \quad\quad \quad\quad\quad \text{by considering each pair of pages $a$ and $b$ once }\\
	&= \sum_{a=1}^n \sum_{b=a+1}^n p_a  Pr_{\pi \sim \xb}[\pi(a) <\pi(b)] + p_bPr_{\pi  \sim \xb}  [\pi(b)  <\pi(a)].
	\end{align*}
\end{proof}

Intuitively by Proposition \ref{prp:1} we can compute the social welfare of a strategy by comparing the ranks of every pairs of webpages. Therefore we define $\vabx$ to be the amount that the pair of webpages $a$ and $b$ adds to the social welfare in strategy $\x$, i.e. $\vabx = p_a  Pr_{\pi \sim \xb}[\pi(a) <\pi(b)] + p_b Pr_{\pi  \sim \xb}  [\pi(b)  <\pi(a)]$. Thus we can rewrite Proposition \ref{prp:1} as $\SW(\x)=\sum_{a=1}^n \sum_{b=a+1}^n \vabx $.
Hence for every strategy $\x$, $$\frac{\SW(\x)}{\SW(\OPT)} = \frac{\sum_{a=1}^n \sum_{b=a+1}^n \vabx}{\sum_{a=1}^n{p_a(n-a)}}.$$
In Lemma \ref{lemghashange} we provide our main tool for bounding the price of competition in the linear ranking duel.

For proving Lemma \ref{lemghashange}, first we need to prove the following lemma.
\begin{lemma}\label{zibaeq}
	For any three integer numbers $n$, $a$, and $k$ such that $1 \leq a \leq n-1$ and $2 \leq k \leq n$, we have
	\begin{equation}
	\sum_{i=0}^{k-1} \binom{a-1}{i} \binom{n-a}{k-i-1}(k-i-1) = (n-a)\binom{n-2}{k-2}
	\end{equation}
\end{lemma}
\begin{proof}
	%Our method for proving the Equation \ref{zibaeq} is double counting. Consider we have $n$ white balls numbered from $1$ to $n$ and we want to color $k-2$ balls with black and one of them with red in such a way that:
	We use double counting. Consider a problem in which we have $n$ balls numbered from $1$ to $n$ and our goal is to color the balls such that
	\begin{itemize}
		\item We have $k-2$ black balls, $1$ red ball, and $n-(k-2)-1$ uncolored balls.
		\item $a$-th ball is uncolored.
		\item Index of the red ball is higher than $a$.
	\end{itemize}
	We calculate the number of different ways that we can color the balls. One way to calculate this is to select $k-1$ balls to color first, and then color $k-2$ of them with black and one of them which has an index higher than $a$ with red. The number of such colorings can be calculated as follows
	\begin{equation*}
	\sum_{i=0}^{k-1} \binom{a-1}{i} \binom{n-a}{k-i-1}(k-i-1),
	\end{equation*}
	which is equal to the left side of Equation \eqref{zibaeq}. The other way to count the number of valid colorings is to first color one ball with rank higher than $a$ with red, and then color $k-2$ balls out of all balls except the $a$-th and the red ball. We have $n-a$ choices for coloring the red ball and $\binom{n-2}{k-2}$ choices for coloring the black ones. Thus we have $(n-a)\binom{n-2}{k-2}$ different ways, which is equal to the right side of Equation \eqref{zibaeq}. Therefore, both sides of Equation \eqref{zibaeq} are equal the number of different valid colorings.
\end{proof}

\begin{lemma}\label{lemghashange}
	\label{fact8-8}
	Given a strategy $\x$, if there exist an integer $k$ such that $2 \leq k \leq n$ and for all $k$ different indices $i_1 < i_2 < \ldots < i_k$,
	%Given a strategy $\x$, if for any $2 \leq k \leq n$ different indices $i_1 < i_2 < \ldots < i_k$,
	\begin{equation*}
	\frac { \sum_{a=1}^k \sum_{b = a+1}^k h_{i_ai_b}(\x)}{\sum_{a=1}^k p_{i_a}(k-a) } \geq \alpha,
	\end{equation*}
	then $\frac{\SW(\x)}{\SW(\OPT)} \geq \alpha$.
\end{lemma}
\begin{proof}
	We compute the summation of 
	\begin{equation*}
	\sum_{a=1}^k \sum_{b = a+1}^k p_{i_a}  Pr_{\pi \sim \xb}[\pi(i_a) <\pi(i_b)] + p_{i_b}  Pr_{\pi \sim \xb}[\pi(i_b) <\pi(i_a)] \geq \alpha(\sum_{a=1}^k p_{i_a}(k-a))
	\end{equation*}
	for all $\binom{n}{k}$ possible indices $1 \leq i_1 < \ldots < i_k \leq n$, therefore we have:
	\begin{equation}\label{eqmeq}
	\sum_{1 \leq i_1  < \ldots < i_k \leq n} \sum_{a=1}^k \sum_{b = a+1}^k p_{i_a}  Pr_{\pi \sim \xb}[\pi(i_a) <\pi(i_b)] + p_{i_b}  Pr_{\pi \sim \xb}[\pi(i_b) <\pi(i_a)] \geq
	\end{equation}
	\begin{equation*}
	\alpha( \sum_{1 \leq i_1 < \ldots < i_k \leq n} \sum_{a=1}^k p_{i_a}(k-a))
	\end{equation*}
	For each $1 \leq a < b \leq k$, $Pr_{\pi \sim \xb}[\pi(i_a) <\pi(i_b)] + p_{i_b}  Pr_{\pi \sim \xb}[\pi(i_b) <\pi(i_a)]$ appears $\binom{n-2}{k-2}$ times in the left side, hence
	\begin{equation*}
	\sum_{1 \leq i_1 < i_2 < i_3 < \ldots < i_k \leq n} \sum_{a=1}^k \sum_{b = a+1}^k p_{i_a}  Pr_{\pi \sim \xb}[\pi(i_a) <\pi(i_b)] + p_{i_b}  Pr_{\pi \sim \xb}[\pi(i_b) <\pi(i_a)] 
	\end{equation*}
	\begin{equation*}
	= \sum_{a=1}^n \sum_{b = a+1}^n \binom{n-2}{k-2}(p_{a}  Pr_{\pi \sim \xb}[\pi(a) <\pi(b)] + p_{b}  Pr_{\pi \sim \xb}[\pi(b) <\pi(a)]).
	\end{equation*}
	Moreover, the coefficient of $p_a$ in the right side of Equation \eqref{eqmeq} is equal to $\sum_{b=0}^{k-2} \binom{a-1}{b} \binom{n-a}{k-1-b} (k-b) $, Therefore
	\begin{equation*}
	\sum_{a=1}^n \sum_{b = a+1}^n \binom{n-2}{k-2}(p_{a}  Pr_{\pi \sim \xb}[\pi(a) <\pi(b)] + p_{b}  Pr_{\pi \sim \xb}[\pi(b) <\pi(a)]) \geq
	\end{equation*}
	\begin{equation*}
	\alpha(\sum_{a=1}^n p_{a} \sum_{b=0}^{k-2} \binom{a-1}{b} \binom{n-a}{k-1-b} (k-b))
	\end{equation*}
	Lemma \ref{zibaeq} states that the right side is equal to $\alpha(\sum_{a=1}^n p_{a} (n-a) \binom{n-2}{k-2})$.
	Thus, by dividing both sides by $\binom{n-2}{k-2}$ we have
	\begin{equation*}
	\sum_{a=1}^n \sum_{b=a+1}^n p_a  Pr_{\pi \sim \xb}[\pi(a) <\pi(b)] + p_b  Pr_{\pi \sim \xb}[\pi(b) <\pi(a)] \geq \alpha (\sum_{a=1}^n p_a(n-a))
	\end{equation*}
	which concludes
	$\frac{\SW(\x)}{\SW(\OPT)} \geq \alpha.$
\end{proof}

Now our goal is to provide a lower bound for $\alpha$ when $\x$ is a minimax strategy. In order to do that, first we provide some structural properties of the minimax strategies. Leveraging these properties we write a mathematical program with $k$ variables $p_a$ and $\binom{k}{2}$ variables $\vab$. Finally, we provide a factor-revealing linear program to obtain a close lower bound for $\alpha$ in the corresponding mathematical program.
%Thus we write mathematical program MP \ref{mp:1} with $k$ variables $p_i$ and $\binom{k}{2}$ variables $\vab$ in order to minimize that $\alpha$. Then we provide some structural properties of the minimax strategies. Leveraging these properties we rewrite the mathematical program MP \ref{mp:2}. Finally, we provide a linear program to obtain a close lower bound for $\alpha$ in LP \ref{mp:4}.

%  \begin{alignat}{2}
%    \text{minimize }   & \alpha  \label{mp:1}\\
%    \text{subject to } & \alpha=\ \frac{\sum_{a=1}^k \sum_{b=a+1}^k \vabx}{\sum_{a=1}^k{p_a(k-a)}} \\
%    					\ & \ p_i \geq 0 &\ & \forall 1\leq i\leq k\\
%                       & \sum_{1 \leq i \leq k}{p_i}  \leq 1
%  \end{alignat}

%Still we need to add constraints regarding the structures of the minimax strategies. 
In Lemmas \ref{lm:s1}, \ref{lem_na_nb}, \ref{fact_6}, and Proposition \ref{fact_5}  we provide the structural properties of the minimax strategies.
\begin{lemma}\label{lm:s1}
Let $\x$ be a minimax strategy and $a$ and $b$ be two webpages such that $p_a \geq p_b$. 
%For any permutation $\piba$ such that $\x_{\piba} > 0$ and $b$ comes before $a$ in $\piba$ i.e. $i = \pi_{ba}(b) < \pi_{ba}(a) = j$, we have
Let $\piba$ be any permutation in the support of $\x$ in which $b$ precedes $a$. Let $i < j$ be the respective position of $a$ and $b$ in $\piba$, then strategy $\x$ must satisfy,
$$Pr_{\pi \sim \xb}[ i < \pi(b) \leq j] + Pr_{\pi \sim \xb}[ i \leq \pi(b) < j] \geq \frac{p_a}{p_b} ( Pr_{\pi \sim \xb}[ i < \pi(a) \leq j] + Pr_{\pi \sim \xb}[ i \leq \pi(a) < j]  ).$$
\end{lemma}
\begin{proof}
	Let $\piab$ be a permutation which is constructed from $\piba$ by swapping $a$ and $b$. This means $\piab(a) =i$ and $\piab(b) = j$. Consider strategy $\x' = \x + \epsilon_{\piab} - \epsilon_{\piba}$ which is produced from $\xb$ by increasing the probability of $\piab$ by $\epsilon$ and decreasing the probability of $\piba$ by $\epsilon$. Since $\xb$ is a minmax strategy, we have $u^A(\xb, \xb') \geq 0$.
	
	As mentioned in Section \ref{sec:2.1},  $u_w^A(\x, \x')=\sum_\omega u_\omega^A(x,\x')$, where $u_\omega^A(\mathbf{x},\mathbf{x'})$ is the payoff of the game when webpage $\omega$ is searched.
	\begin{equation}
	u(\mathbf{x},\mathbf{x'}) = \sum_{\omega=1}^n p_\omega u_\omega^A(\mathbf{x},\mathbf{x'})
	\end{equation}
	Since, both players act the same for all pages except $a$ and $b$, we have
	\begin{equation}
	\forall 1 \leq \omega \leq n, \omega \neq a, \omega \neq b \hspace{1cm} u_\omega^A(\mathbf{x},\mathbf{x'}) = 0.
	\end{equation}
	Therefore
	\begin{equation}
	\label{eqn_ab}
	u^A(\mathbf{x},\mathbf{x'}) = p_a u_a^A(\mathbf{x},\mathbf{x'}) + p_b u_b^A(\mathbf{x},\mathbf{x'}).
	\end{equation}
	
	%Let $\xb'' = \frac{1}{1-\epsilon}(\xb - \epsilon_{\pi'}) = \frac{1}{1-\epsilon}(\xb' - \epsilon_{\pi''})$.
	Let $\x''$ be a new strategy such that $\x''_{\piba} = \x_{\piba}-\epsilon$, and for every other permutation $\pi$, $\x''_\pi = \frac{1}{1-\epsilon} \x_\pi$. In fact strategy $\x$ plays strategy $\x''$ with probability $1-\epsilon$ and plays permutation $\piba$ with probability $\epsilon$. Also strategy $\x'$ plays $\x''$ with probability $1-\epsilon$ and plays permutation $\piab$ with probability $\epsilon$. Therefore we can compute $u_a^A(\x, \x')$ as follows
	%Assume permutation $\pi_1$ is randomly drawn from strategy $\xb$. This process can be simulated by setting $\pi_1 = \pi'$ with probability $\epsilon$ and setting $\pi_1$ with respect to strategy $\xb''$ with probability $1-\epsilon$. The same holds for permutation $\pi_2$ which is randomly drawn from strategy $\xb'$. Therefore, we write $u_a^A(\mathbf{x},\mathbf{x'})$ as follows:
	
	\begin{align*}
	u_a^A(\mathbf{x},\mathbf{x'}) 
	&= (1-\epsilon)^2(Pr_{\pi_1 \sim \xb'', \pi_2 \sim \xb''}[\pi_1(a) < \pi_2(a)] - Pr_{\pi_1 \sim \xb'', \pi_2 \sim \xb''}[\pi_1(a) > \pi_2(a)]) \\ 
	&\omit\hfill\text{The case that both strategies play $\x''$} \\
	&+ \epsilon(1-\epsilon) (Pr_{\pi_2 \sim \xb''}[\piba(a) < \pi_2(a)] - Pr_{\pi_2 \sim \xb''}[\piba(a) > \pi_2(a)]) \\ 
	&\omit\hfill\text{The case that $\x$ plays $\piba$ and $\x'$ plays $\x''$} \\
	&+ \epsilon(1-\epsilon) (Pr_{\pi_1 \sim \xb''}[\pi_1(a) < \piab(a)] - Pr_{\pi_1 \sim \xb''}[\pi_1(a) > \piab(a)]) \\
	&\omit\hfill\text{The case that $\x$ plays $\x''$ and $\x'$ plays $\piab$} \\
	&+ \epsilon^2 (Pr[\piba(a) < \piab(a)] - Pr[\piba(a) > \piab(a)])\\
	&\omit\hfill\text{The case that $\x$ plays $\piba$ and $\x'$ plays $\piab$} 
	\end{align*}
	Consider the first term of the above equation. Since both permutations $\pi_1$ and $\pi_2$ are drawn from the same distribution $\xb''$, this term is zero. On the other hand, recall that $\piab(a)=\piba(b)=i$ and $\piab(b)=\piba(a)=j$. Thus $Pr[\piba(a) < \piab(a)] - Pr[\piba(a) > \piab(a)] = Pr[j<i] - Pr[j>i] = -1$, hence
	\begin{eqnarray}\label{eqn_a1}
	u_a^A(\mathbf{x},\mathbf{x'}) 
	&=& \epsilon (1-\epsilon) (Pr_{\pi_2 \sim \xb''}[j < \pi_2(a)] - Pr_{\pi_2 \sim \xb''}[j > \pi_2(a)]) \\ \nonumber 
	&+& \epsilon (1-\epsilon) (Pr_{\pi_1 \sim \xb''}[\pi_1(a) < i] - Pr_{\pi_1 \sim \xb''}[\pi_1(a) > i])  \\ \nonumber
	&-& \epsilon^2 
	\end{eqnarray}
	Similarly we have %We write $u^A_b(\xb, \xb')$ is the same way as follows:
	\begin{eqnarray}\label{eqn_b13}
	u_b^A(\mathbf{x},\mathbf{x'}) 
	&=&\epsilon (1-\epsilon) (Pr_{\pi_2 \sim \xb''}[i < \pi_2(b)] - Pr_{\pi_2 \sim \xb''}[i > \pi_2(b)]) \\ \nonumber 
	&+& \epsilon (1-\epsilon) (Pr_{\pi_1 \sim \xb''}[\pi_1(b) < j] - Pr_{\pi_1 \sim \xb''}[\pi_1(b) > j]) \\ \nonumber
	&+& \epsilon^2 
	\end{eqnarray}
	$u^A(\xb, \xb') \geq 0$ thus by Equation \eqref{eqn_ab} $p_au_a^A(\mathbf{x},\mathbf{x'})  + p_bu_b^A(\mathbf{x},\mathbf{x'})  \geq 0$
	. Hence using Equations \eqref{eqn_a1} and \eqref{eqn_b13} we have
	\begin{eqnarray}
	\nonumber && p_a\epsilon (1-\epsilon) (Pr_{\pi_2 \sim \xb''}[j < \pi_2(a)] - Pr_{\pi_2 \sim \xb''}[j > \pi_2(a)]) \\ \nonumber &+& p_a\epsilon (1-\epsilon) (Pr_{\pi_1 \sim \xb''}[\pi_1(a) < i] - Pr_{\pi_1 \sim \xb''}[\pi_1(a) > i]) - p_a \epsilon^2 \\ \nonumber
	&+&
	p_b\epsilon(1-\epsilon) (Pr_{\pi_2 \sim \xb''}[i < \pi_2(b)] - Pr_{\pi_2 \sim \xb''}[i > \pi_2(b)]) \\ \nonumber &+& p_b\epsilon (1-\epsilon) (Pr_{\pi_1 \sim \xb''}[\pi_1(b) < j] - Pr_{\pi_1 \sim \xb''}[\pi_1(b) > j]) + 
	p_b \epsilon^2 \geq 0
	\end{eqnarray}
	Note that $p_a \geq p_b$ which means $-p_a \epsilon^2 + p_b \epsilon^2 \leq 0$. Thus we conclude:
	\begin{eqnarray}
	{p_a}(Pr_{\pi_2 \sim \xb''}[j < \pi_2(a)] &-& Pr_{\pi_2 \sim \xb''}[j > \pi_2(a)])  \nonumber \\ + {p_a}(Pr_{\pi_1 \sim \xb''}[\pi_1(a) < i] &-& Pr_{\pi_1 \sim \xb''}[\pi_1(a) > i]) 
	\nonumber \\ 
	+ {p_b}(Pr_{\pi_2 \sim \xb''}[i < \pi_2(b)]  &-&   Pr_{\pi_2 \sim \xb''}[i > \pi_2(b)]) \nonumber \\ + {p_b}(Pr_{\pi_1 \sim \xb''}[\pi_1(b) < j] &-& Pr_{\pi_1 \sim \xb''}[\pi_1(b) > j]) 
	\geq 0\label{eqn_last1}
	\end{eqnarray}
	Note that as $\epsilon$ approaches zero, strategy $\xb''$ approaches $\xb$. Hence, we write Equation \eqref{eqn_last1} as follows:
	\begin{eqnarray}
	{p_a}(Pr_{\pi \sim \xb}[j < \pi(a)] &-& Pr_{\pi \sim \xb}[j > \pi(a)])  \nonumber \\ + {p_a}(Pr_{\pi \sim \xb}[\pi(a) < i] &-& Pr_{\pi \sim \xb}[\pi(a) > i]) 
	\nonumber \\ 
	+ {p_b}(Pr_{\pi \sim \xb}[i < \pi(b)]  &-&   Pr_{\pi \sim \xb}[i > \pi(b)]) \nonumber \\ + {p_b}(Pr_{\pi \sim \xb}[\pi(b) < j] &-& Pr_{\pi \sim \xb}[\pi(b) > j]) 
	\geq 0\label{eqn_last2}
	\end{eqnarray}
	Rearranging the terms we have $$Pr_{\pi \sim \xb}[ i < \pi(b) \leq j] + Pr_{\pi \sim \xb}[ i \leq \pi(b) < j] \geq \frac{p_a}{p_b} ( Pr_{\pi \sim \xb}[ i < \pi(a) \leq j] + Pr_{\pi \sim \xb}[ i \leq \pi(a) < j]  ).$$
\end{proof}

Intuitively Lemma \ref{lm:s1} shows that if $p_a \geq p_b$ and there is a permutation in which $b$ comes before $a$, then the probability that $\x$ ranks $b$ in interval $[i, j]$ (counting the non-endpoint elements twice) is greater than the probability that $\x$ ranks $a$ in this interval by a factor of $\frac{p_a}{p_b}$. Otherwise, by swapping the rank of $a$ and $b$ we can achieve a strategy that beats $\x$.

%Now by Lemma \ref{lm:s1} we can prove the following lemma
\begin{lemma}
\label{lem_na_nb}
Let $\xb$ be a minimax strategy and $\x_{\pi}$ be the probability that strategy $\xb$ plays permutation $\pi$. For every pair of webpages $a$ and $b$ with $p_a \geq p_b$, we have
\begin{equation}\label{eq_na_nb2}
Pr_{\pi \sim \xb}[\pi(a) < \pi(b)] \geq (\frac{p_a}{2p_b}-1)Pr_{\pi \sim \xb}[\pi(b)< \pi(a)].
\end{equation} 
\end{lemma}
\begin{proof}
	\input{omit.tex}
\end{proof}

Briefly, in the proof of Lemma \ref{lem_na_nb} we propose an algorithm to find a set of permutations $\Pi$ in $\x$, such that 1) for each $\pi \in \Pi$, $b$ comes before $a$, 2)  for each permutation $\pi'$ in $\x$ in which $b$ comes before $a$, there is a permutation $\pi \in \Pi$, such that $\pi(b) \leq \pi'(a) \leq \pi(a)$, and 3) the interval of the ranks of $b$ and $a$ are distinct, i.e. for two permutations $\pi, \pi' \in \Pi, [\pi(b), \pi(a)] \cap [\pi'(b), \pi'(a)] = \emptyset$. We apply the inequality in Lemma \ref{lm:s1} for all permutations in $\Pi$ to achieve Lemma \ref{lem_na_nb}.

In Proposition \ref{fact_5} and Lemma \ref{fact_6} we provide lower bounds for $\vabx$ when $\x$ is a minimax strategy. Hence we can use these lower bounds in the proposed mathematical program to achieve a lower bound for the $\POC$.
\begin{proposition}
\label{fact_5}
For minimax strategy $\xb$ and webpages $a$ and $b$ such that $p_a \geq p_b$,
$\vabx \geq p_b$.
\end{proposition}
\begin{proof}
	Since $p_a \geq p_b$ and $Pr_{\pi \sim \xb}[\pi(a) <\pi(b)] + Pr_{\pi \sim \xb}[\pi(b) <\pi(a)]=1$, we have:
	\begin{eqnarray*}
		p_a  Pr_{\pi \sim \xb}[\pi(a) <\pi(b)] + p_b  Pr_{\pi \sim \xb}[\pi(b) <\pi(a)]  &\geq & p_b  Pr_{\pi \sim \xb}[\pi(a) <\pi(b)] + p_b Pr_{\pi \sim \xb}[\pi(b) <\pi(a)]  \\
		&=& p_b
	\end{eqnarray*} 
\end{proof}

\begin{lemma}
\label{fact_6}
For minimax strategy $\xb$ and webpages $a$ and $b$ such that $p_a \geq p_b$,
$\vabx \geq p_a - 2p_b + \frac{2p_b^2}{p_a}$.
\end{lemma}
\begin{proof}
	%We use Lemma \ref{lem_na_nb} and the fact that $Pr_{\pi \sim \xb}[\pi(a) <\pi(b)] + Pr_{\pi \sim \xb}[\pi(b) <\pi(a)]=1$ to bound $Pr_{\pi \sim \xb}[\pi(a) <\pi(b)]$. 
	First we claim that ${Pr_{\pi \sim \xb}[\pi(a) <\pi(b)]}  \geq 1- \frac{2p_b}{p_a}$. By Lemma \ref{lem_na_nb}, 
	\begin{align}
	Pr_{\pi \sim \xb}[\pi(a) <\pi(b)] & \geq (\frac{p_a}{2p_b}-1)  Pr_{\pi \sim \xb}[\pi(b) <\pi(a)] \quad\quad\text{dividing the sides by $Pr_{\pi \sim \xb}[\pi(b) <\pi(a)]$}\nonumber  \\ \Rightarrow
	\frac{Pr_{\pi \sim \xb}[\pi(a) <\pi(b)]}{Pr_{\pi \sim \xb}[\pi(b) <\pi(a)]} & \geq (\frac{p_a}{2p_b}-1)  \nonumber  \\ \Rightarrow
	\frac{1}{Pr_{\pi \sim \xb}[\pi(b) <\pi(a)]} &\geq \frac{p_a}{2p_b} \nonumber  \\\Rightarrow
	{Pr_{\pi \sim \xb}[\pi(b) <\pi(a)]} &\leq \frac{2p_b}{p_a} \quad\quad\quad\quad\quad\quad\quad\quad \text{since $ {Pr_{\pi \sim \xb}[\pi(b) <\pi(a)]} = {1-Pr_{\pi \sim \xb}[\pi(a) <\pi(b)]}$}\nonumber \\\Rightarrow
	{Pr_{\pi \sim \xb}[\pi(a) <\pi(b)]} &\geq 1- \frac{2p_b}{p_a} \label{eqn_6}
	\end{align}
	Now we use Inequality \eqref{eqn_6} and the fact that $p_a \geq p_b$ as follows
	\begin{eqnarray*}
		&& p_a  Pr_{\pi \sim \xb}[\pi(a) <\pi(b)] + p_b  Pr_{\pi \sim \xb}[\pi(b) <\pi(a)]  \\&=& p_a  Pr_{\pi \sim \xb}[\pi(a) <\pi(b)] + p_b  (1-Pr_{\pi \sim \xb}[\pi(a) <\pi(b)])\\ &=& p_b + (p_a-p_b) Pr_{\pi \sim \xb}[\pi(a) <\pi(b)] \quad\quad\quad\quad\quad\quad\quad\quad\quad\quad\quad\quad\quad\quad\quad\quad\quad\quad\text{By Inequality \eqref{eqn_6}}\\
		&\geq& p_b + (p_a-p_b) (1-\frac{2p_b}{p_a})= p_a-2p_b+\frac{2p_b^2}{p_a}
	\end{eqnarray*} 
\end{proof}

Leveraging the properties of the minimax strategies we write MP \ref{mp:2}. In MP \ref{mp:2}, Constraints \ref{cnd1} and \ref{cnd2} force $p_a$'s to satisfy the probability constraints. Using Proposition \ref{fact_5}, Constraint \ref{cnd3} forces $\vab$ to be no less than $p_b$ and due to Lemma \ref{fact_6}, Constraint \ref{cnd3} forces $\vab$ to be no less than $p_a - 2p_b + \frac{2p^2_b}{p_a}$. By Lemma \ref{lemghashange}, $\alpha$ in Constraint \ref{cnd0} gives a lower bound for the $\POC$.
  \begin{alignat}{3}
    \text{minimize }  \ & \alpha  \label{mp:2}\\
    \text{subject to } \ & \alpha=\frac{\sum_{a=1}^k \sum_{b=a+1}^k \vab}{\sum_{a=1}^k{p_a(k-a)}} \label{cnd0}\\
    					& p_a \geq 0 &\ & \forall 1\leq a\leq k \label{cnd1}\\ 
                       & \sum_{1 \leq a \leq k}{p_a}  \leq 1 \label{cnd2}\\ 
                       & \vab \geq p_b &\ & \forall 1\leq a < b \leq k \label{cnd3}\\ 
                       & \vab \geq p_a - 2p_b + \frac{2p^2_b}{p_a} &\ & \forall 1\leq a < b \leq k \label{cnd4}
  \end{alignat}

For each $k$, let $\alpha_k$ be the optimal value of the objective function in MP \ref{mp:2}.
\begin{lemma}\label{lm:mp}
$\alpha_k$ is a lower bound for the $\POC$ of the linear ranking duel where $n \geq k$.
\end{lemma}
\begin{proof}
	As we have in Lemma \ref{lemghashange}, in a strategy $\x$, if for every $k$ indices $i_1 < i_2 < \ldots < i_k$ we have $$\frac { \sum_{a=1}^k \sum_{b = a+1}^k p_{i_a}  Pr_{\pi \sim \xb}[\pi(i_a) <\pi(i_b)] + p_{i_b}  Pr_{\pi \sim \xb}[\pi(i_b) <\pi(i_a)] } {\sum_{a=1}^k p_{i_a}(k-a) } \geq \alpha,$$ then $\frac{\SW(\x)}{\SW(\OPT)} \geq \alpha$.
	
	Now consider $k$ indices which minimizes this fraction and for $1\leq i_a \leq k$, let $p_a$ be the probability of that webpage and for $1\leq i_a < i_b \leq k$ let $\vab$ be $p_a  Pr_{\pi \sim \xb}[\pi(i_a) <\pi(i_b)] + p_b Pr_{\pi  \sim \xb}  [\pi(i_b)  <\pi(i_a)]$. Constraint \ref{cnd1} forces the probabilities to be positive and Constraint \ref{cnd2} forces the sum of probabilities to be less than $1$. Since we are considering minmax strategies, we use Proposition \ref{fact_5} and Lemma \ref{fact_6} to give lower bounds in Constraints \ref{cnd3} and \ref{cnd4} for $\vab$. Therefore any minmax strategy in a ranking duel is a feasible solution for MP \ref{mp:2}. Thus $\alpha$ is a lower bound for the \POC.
\end{proof}

In Theorem \ref{trm:asli} we formally prove $\alpha_{10} \geq \POCBound$, which results in $\POC \geq \POCBound$ for any ranking duel with $n\geq 10$ webpages. Moreover, we write a computer program to find $\alpha_k$ for $2\leq k \leq 100$ (see Figure \ref{fig_expr}).

In order to prove Theorem \ref{trm:asli}, first we should prove the following lemma.
\begin{lemma}\label{khat}
	For $0 \leq p_b \leq p_a \leq 1$, $max\{p_b, p_a-2p_b+\frac{2p_b^2)}{p1}\} \geq max\{p_b,p_a-2p_b, \frac{(p_a-p_b)}{1.208},\frac{(2p_a-*p_b)}{3.2}\}$.
\end{lemma}
\begin{proof}
	If $p_b=0$ the inequality holds, thus we assume $p_b > 0$. Let $z=\frac{p_a}{p_b}$. We multiply each side by $\frac{1}{p_b}$. Thus we should prove 
	\begin{align*}
	&\max\{ \frac{1}{p_b}p_b, \frac{1}{p_b}(p_a-2p_b)+\frac{1}{p_b}\frac{2p_b^2}{p_a}\} \geq max\{ \frac{1}{p_b}p_b, \frac{1}{p_b}(p_a-2p_b),  \frac{1}{p_b}\frac{p_a-p_b}{1.208}, \frac{1}{p_b}\frac{2p_a-p_b}{3.2}\} \\
	\Rightarrow&\max\{1, z-2+\frac{2}{z}\} \geq \max\{1, z-2, \frac{z-1}{1.208}, \frac{2z-1}{3.2}\}
	\end{align*}
	$z-2 + \frac{2}{z} \geq z-2$, thus it is sufficient to prove $z-2+\frac{2}{z} \geq \frac{z-1}{1.208}$ and $z-2+\frac{2}{z} \geq \frac{2z-1}{3.2}$. Since $(2\times 1.208 -1)^2 - 8(1.208-1)\times1.208 > 0$, $(1.208-1)z^2 - z(2\times1.208-1) +2\times1.208 >0$. Thus by dividing the terms by $1.208z$ we have 
	\begin{align*}
	&z - 2 + \frac{2}{z} - \frac{z-1}{1.208} > 0\\
	\Rightarrow&z-2+\frac{2}{z} > \frac{z-1}{1.208}.
	\end{align*}
	
	Moreover, since $(2\times 3.2 -1)^2 - 8(3.2-2)\times3.2 > 0$, $(3.2-2)z^3.2 - z(2\times3.2-1) +2\times3.2 >0$. Thus by dividing the terms by $3.2z$ we have 
	\begin{align*}
	&z - 2 + \frac{2}{z} - \frac{2z-1}{3.2} > 0\\
	\Rightarrow& z-2+\frac{2}{z} > \frac{2z-1}{3.2}.
	\end{align*}
\end{proof}

\begin{theorem}\label{trm:asli}
	For a linear ranking duel with $n \geq 10$ webpages, $\POC \geq \POCBound$.
\end{theorem}
\begin{proof}
	First we try to replace Constraint \ref{cnd4} with some linear constraints. More precisely we replace $\max\{p_b, p_a-2p_b+\frac{2p_b^2)}{p1}\}$ by the maximum of four linear terms. Lemma \ref{khat} shows we can change the MP \ref{mp:2} to achieve the following program.
	\begin{alignat}{3}
	\text{minimize }  \ &  \alpha  \label{mp:3}\\
	\text{subject to } \ & \alpha=\frac{\sum_{a=1}^k \sum_{b=a+1}^k \vab}{\sum_{a=1}^k{p_a(k-a)}} \label{cnd5}\\
	& \ p_a \geq 0 &\ & \forall 1\leq a\leq k\\
	& \sum_{1 \leq a \leq k}{p_a}  \leq 1\\
	& \vab \geq p_b &\ & \forall 1\leq a < b \leq k\\
	& \vab \geq p_a - 2p_b  &\ & \forall 1\leq a < b \leq k\\
	& \vab \geq \frac{p_a-p_b}{1.208} &\ & \forall 1\leq a < b \leq k\\
	& \vab \geq \frac{2p_a-p_b}{3.2}  &\ & \forall 1\leq a < b \leq k
	\end{alignat}
	
	Still Constraint \ref{cnd5} is not linear. We scale the probabilities such that ${\sum_{a=1}^k{p_a(k-a)}}$ becomes equal to $1$, hence we can have a linear constraint instead of Constraint \ref{cnd5}.
	Let  $(\alpha, V=\langle v_{12}, v_{13}, \ldots, v_{k-1k}\rangle, P=\langle p_1, p_2, \ldots, p_k\rangle)$ be a feasible solution for MP \ref{cnd5}. We provide a feasible solution with variables $(\alpha', V'=\langle v'_{12}, v'_{13}, \ldots, v'_{k-1k}\rangle, P=\langle p'_1, p'_2, \ldots, p'_k\rangle)$ such that $\alpha' = \alpha$.
	Let $c=\frac{1}{\sum_{a=1}^k{p_a(k-a)}}$. Let $p_i' = cp_i$ and $\vab'=c\vab$. Hence,
	\begin{align*}
	\alpha' =& \sum_{a=1}^k \sum_{b=a+1}^k \vab' \\
	=& \sum_{a=1}^k \sum_{b=a+1}^k c\vab\\
	=& \frac{\sum_{a=1}^k \sum_{b=a+1}^k \vab}{\sum_{a=1}^k p_a(k-a)}\\
	=& \alpha.
	\end{align*}
	Moreover, all Constraints \ref{c1}, \ref{c2}, \ref{c3}, \ref{c4}, and \ref{c5} hold since by dividing the sides of inequality we achieve the constraints in MP \ref{mp:3}. Thus for any feasible solution for MP \ref{mp:3} there is a feasible solution for LP \ref{mp:4} with $\alpha'=\alpha$.
	\begin{alignat}{3}
	\text{minimize }  \ & \alpha'  \label{mp:4}\\
	\text{subject to } \ & \alpha' = \sum_{a=1}^k \sum_{b=a+1}^k \vab'\\
	& p'_a \geq 0 &\ & \forall 1\leq a\leq k \label{c1}\\ 
	& \sum_{a=1}^k{p'_a(k-a)}=1 \label{c1.5}\\ 
	& \vab' \geq  p'_b &\ & \forall 1\leq a < b \leq k \label{c2}\\ 
	& \vab' \geq p'_a - 2p'_b  &\ & \forall 1\leq a < b \leq k \label{c3}\\
	& \vab' \geq \frac{p'_a-p'_b}{1.208}  &\ & \forall 1\leq a < b \leq k \label{c4}\\ 
	& \vab' \geq \frac{2p'_a-p'_b}{3.2}  &\ & \forall 1\leq a < b \leq k \label{c5}
	\end{alignat}
	Now we write the dual program of LP \ref{mp:4}. Finding a feasible solution for the dual program, we provide a lower bound for $\alpha$.
	We can set the objective function of LP \ref{mp:4} to the minimization of $\sum_{a=1}^k \sum_{b=a+1}^k \vab'$ and remove Constraint \ref{c1}. Then by assigning variables $\theta, \beta_{i, j}, \gamma_{i, j}, \lambda_{i, j}, \rho_{i, j}$ to Constraints \ref{c1.5}, \ref{c2}, \ref{c3}, \ref{c4}, and \ref{c5} respectively, the dual LP is as follows
	\begin{alignat}{3}
	\text{maximize }  \ & \theta \label{mp:5}\\
	\text{subject to } \ & \theta(k-a) -\sum_{i=1}^{a-1}{\beta_{i, a}}\\
	&-\sum_{j=a+1}^{k}{\gamma_{a, j}}	+ \sum_{i=1}^{a-1}{2\gamma_{i, a}}\\
	&-\sum_{j=a+1}^{k}{\frac{1}{1.208}\lambda_{a, j}}	+ \sum_{i=1}^{a-1}{\frac{1}{1.208}\lambda_{i, a}}\\
	&-\sum_{j=a+1}^{k}{\frac{2}{3.2}\rho_{a, j}}	+ \sum_{i=1}^{a-1}{\frac{1}{3.2}\rho_{i, a}} \leq 0 &\quad\quad & \forall 1 \leq a \leq k \\
	&\beta_{a, b} + \gamma_{a, b} + \lambda_{a, b} + \rho_{a, b} \leq 1 &\ & \forall 1 \leq a < b \leq k\\
	& \beta_{a, b}, \gamma_{a, b}, \lambda_{a, b}, \rho_{a, b} \geq 0 &\ & \forall 1 \leq a < b \leq k
	%	& p'_a \geq 0 &\ & \forall 1\leq a\leq k \label{c1}\\ 
	\end{alignat}
	Table \ref{tablemastak} provides a feasible solution for LP \ref{mp:5} with $k=10$, in which $\theta=\POCBound$. Thus in LP \ref{mp:4}, $\alpha' \geq 0.61$ and therefore $\POC \geq \POCBound$.
\end{proof}

\input{plot.tex}

%% file: omit.tex
%Let $N$ be the set of all permutations such as $\pi$ with $\xb_{\pi} > 0$, and let $S$ be the set of all permutations in $N$ with $\pi(b) < \pi(a)$. We select a subset of $S$ via the following algorithm:
%\begin{enumerate}
%\item Let $S^* = S$, and $\hat{S} = \emptyset$.
%\item Let $\pi \in S^*$ be the permutation with the rightmost $\pi(a)$ among all permutations in $S^*$. 
%\item Let $\Pi = \Pi \cup \pi$.
%\item Let $S^* = S^* - \{ \pi' \in S^*|(\pi'(b), \pi'(a)) \cap (\pi(b), \pi(a)) \neq \emptyset \}$.
%\item if $S^* = \emptyset$ return $\Pi$. Otherwise go to line $2$. 
%\end{enumerate}
First we find the set $\Pi$ using Algorithm \ref{alg1}. Let $N_{ab}$ be the set of all strategies in $\x$ in which $a$ comes before $b$. Similarly let $N_{ba}$ be the set of strategies in $\x$ in which $b$ comes before $a$. In Algorithm \ref{alg1}, initially we consider a set $S^*$ to be $N_{ba}$. Then, at each step we add permutation $\pi \in S^*$ with the rightmost $a$ to $\Pi$, and remove all permutations $\pi'$ such that the interval $[\pi'(b), \pi'(a)]$ overlaps the interval $[\pi(b), \pi(a)]$ from $S^*$. We repeat this process until all permutations are removed from $S^*$.
\input{alg.tex}

Afterwards for each $\pi \in \Pi$, we apply Lemma \ref{lm:s1} and add all these inequalities to reach the following inequality:
\begin{align}
\label{eqn_sum_pa_pb}
&\sum_{\pi\in \Pi}  Pr_{\pi' \sim \xb}[ \pi(b) < \pi'(b) \leq \pi(a)] + Pr_{\pi' \sim \xb}[ \pi(b) \leq \pi'(b) < \pi(a)] \nonumber \\ 
\geq &\sum_{\pi\in \Pi}\frac{p_a}{p_b} ( Pr_{\pi' \sim \xb}[ \pi(b) < \pi'(a) \leq \pi(a)] + Pr_{\pi' \sim \xb}[ \pi(b) \leq \pi'(a) < \pi(b)]  ). \end{align}
We can partition each term of Inequality \eqref{eqn_sum_pa_pb} into two terms such that in one of them $\pi'(a) > \pi'(b)$ and in the other one  $\pi'(b) > \pi'(a)$. Thus we can rewrite Inequality \eqref{eqn_sum_pa_pb} as follows
\begin{align}
\label{eqn_sum_pa_pb2}
\sum_{\pi\in \Pi}  & Pr_{\pi' \sim \xb}[ \pi'(a) < \pi'(b) \wedge \pi(b) < \pi'(b) \leq \pi(a)]\\
							+& Pr_{\pi' \sim \xb}[ \pi'(b) < \pi'(a) \wedge \pi(b) < \pi'(b) \leq \pi(a)] \nonumber\\
							+& Pr_{\pi' \sim \xb}[ \pi'(a) < \pi'(b) \wedge \pi(b) \leq \pi'(b) < \pi(a)]\nonumber\\
							+& Pr_{\pi' \sim \xb}[ \pi'(b) < \pi'(a) \wedge \pi(b) \leq \pi'(b) < \pi(a)] \nonumber\\
\geq \frac{p_a}{p_b}\sum_{\pi \in \Pi } & Pr_{\pi' \sim \xb}[ \pi'(a) < \pi'(b) \wedge \pi(b) < \pi'(a) \leq \pi(a)]\nonumber\\
				+& Pr_{\pi' \sim \xb}[ \pi'(b) < \pi'(a) \wedge \pi(b) < \pi'(a) \leq \pi(a)] \nonumber\nonumber\\
															+& Pr_{\pi' \sim \xb}[\pi'(a) < \pi'(b) \wedge  \pi(b) \leq \pi'(a) < \pi(b)]\nonumber\\
															+& Pr_{\pi' \sim \xb}[\pi'(a) < \pi'(b) \wedge  \pi(b) \leq \pi'(a) < \pi(b)].\nonumber
 \end{align}
%For the sake of convenience, we define $\rabb$ to denote the term in which $%, $\rbab$, $\raba$, and $\rbaa$ as follows:
%\begin{eqnarray}
%\rabb &=& 
%\sum_{\pi\in \Pi}  Pr_{\pi' \sim \xb}[\pi'(a)< \pi'(b)\wedge \pi(b) < \pi'(b) \leq \pi(a)] \nonumber \\ &+& Pr_{\pi' \sim \xb}[\pi'(a)< %\pi'(b)\wedge \pi(b) \leq \pi'(b) < \pi(a)] \nonumber \\ 
%\rbab &=& 
%\sum_{\pi\in \Pi}  Pr_{\pi' \sim \xb}[\pi'(b)< \pi'(a)\wedge \pi(b) < \pi'(b) \leq \pi(a)] \nonumber \\ &+& Pr_{\pi' \sim \xb}[\pi'(b)< %\pi'(a)\wedge \pi(b) \leq \pi'(b) < \pi(a)] \nonumber \\ 
%\raba &=& 
%\sum_{\pi\in \Pi}  Pr_{\pi' \sim \xb}[\pi'(a)< \pi'(b)\wedge \pi(b) < \pi'(a) \leq \pi(a)] \nonumber \\ &+& Pr_{\pi' \sim \xb}[\pi'(a)< %\pi'(b)\wedge \pi(b) \leq \pi'(a) < \pi(a)] \nonumber \\ 
%\rbaa &=& 
%\sum_{\pi\in \Pi}  Pr_{\pi' \sim \xb}[\pi'(b)< \pi'(a)\wedge \pi(b) < \pi'(a) \leq \pi(a)] \nonumber \\ &+& Pr_{\pi' \sim \xb}[\pi'(b)< %\pi'(a)\wedge \pi(b) \leq \pi'(a) < \pi(a)] \nonumber
%\end{eqnarray}
Now for the sake of convenience, for webpage $c \in \{a, b\}$, we define $R_c^{ab}$ and $R_c^{ba}$ as follows.
$$R_c^{ab}=
\sum_{\pi\in \Pi}  Pr_{\pi' \sim \xb}[\pi'(a)< \pi'(b)\wedge \pi(b) < \pi'(c) \leq \pi(a)] +
Pr_{\pi' \sim \xb}[\pi'(a)< \pi'(b)\wedge \pi(b) \leq \pi'(c) < \pi(a)].$$
$$R_c^{ba}=
\sum_{\pi\in \Pi}  Pr_{\pi' \sim \xb}[\pi'(b)< \pi'(a)\wedge \pi(b) < \pi'(c) \leq \pi(a)] +
Pr_{\pi' \sim \xb}[\pi'(b)< \pi'(a)\wedge \pi(b) \leq \pi'(c) < \pi(a)].$$
%Now for the sake of convenience we define $R^{cd}_(\pi)$ to be the probability that in a randomly drawn permutation $\pi'$ from strategy $\x$, $c$ comes before $d$ and rank of page $e$ is between $\pi(c)$ and $\pi(d)$. Moreover we define $R^{cd}_e = \sum_{\pi \in \Pi} R^{cd}_e(\pi).$
%For example 
%$$\rabb =
%\sum_{\pi\in \Pi}  Pr_{\pi' \sim \xb}[\pi'(a)< \pi'(b)\wedge \pi(b) < \pi'(b) \leq \pi(a)].$$
Thus we can rewrite Inequality \eqref{eqn_sum_pa_pb2} based on the above definitions, and conclude
\begin{eqnarray}
\rabb+ \rbab \geq \frac{p_a}{p_b}(\raba+ \rbaa). \label{eqn_R}
\end{eqnarray}
Note that by the definition of $R$, for each permutation $\pi' \in \Pi$ a randomly drawn permutation $\pi'$ from $\x$ and a webpage $c \in \{a, b\}$, we add $2$ to $R_c$ if $\pi'(c)$ is between the rank of $a$ and $b$ in $\pi$, and we add $1$ if $\pi'(c)$ is equal to the rank of either $a$ or $b$ in $\pi$. More formally we define $r_c(\pi, \pi')$ as follows
\begin{equation}\nonumber
r_c(\pi', \pi)=
\begin{cases}
2 & \text{if } \pi(a) < \pi'(c) < \pi(b) \text{ or } \pi(b) < \pi'(c) < \pi(a)\\%\pi'(c)\text{ is between }\pi(a) \text{ and }\pi(b)  \\
1 & \text{if } \pi'(c) = \pi(a) \text{ or } \pi'(c) = \pi(b)\\
0 & \text{if } \text{otherwise }\\
\end{cases}
\end{equation}
Thus we can write $R^{ab}_c$ and $R^{ba}_c$ as follows%$\rbab$ and $\rbaa$ as follows
\begin{eqnarray}
R^{ab}_c &=& 
\sum_{\pi\in \Pi}  \sum_{\pi'\in {N_{ab}}} \xb_{\pi'} r_c(\pi', \pi). \label{eqqqq1}\\
R^{ba}_c &=& 
\sum_{\pi\in \Pi}  \sum_{\pi'\in {N_{ba}}} \xb_{\pi'} r_c(\pi', \pi). \nonumber
\end{eqnarray}
% Before proving the lemma, we have to prove the following claims.
%\begin{claim}
%\label{fact_1}
%$\rbab \leq 2 \rbaa$.
%\end{claim}
%\begin{proof}[Proof of Fact \ref{fact_1}]
%Note that $w^b(\pi', \pi) \leq 2$ based on the definition of $w^b(\pi', \pi)$.  
%Regarding the construction of $\Pi$, for any rank $i$, if there is a permutation in $\pi^* \in N_{ba}$ such that $\pi^*(a)=i$ or $\pi^*(b)=i$, then there is exactly one permutation $\pi \in \Pi$ such that $i \in [\pi(b), \pi(a)]$. Hence for a permutation $\pi' \in N_{ba}$, there is exactly one permutation $\pi \in \Pi$ such that $r_a(\pi', \pi)$ is non-zero and similarly exactly one permutation with non-zero $r_b(\pi', \pi)$. Thus, for a permutation $\pi' \in N_{ba}$, $\sum_{\pi \in \Pi}\x_{\pi'}r_a(\pi', \pi) \geq 1$ and $\sum_{\pi \in \Pi}\x_{\pi'}r_b(\pi', \pi) \leq 2$. Using this property we prove the claim as follows. By Equation \ref{eqqqq1},
Regarding the construction of $\Pi$, for any rank $i$, there is at most one permutation $\pi \in \Pi$ such that $i$ is between $\pi(a)$ and $\pi(b)$. Thus for a permutation $\pi'$, there is at most one permutation $\pi \in \Pi$ such that $r_b(\pi', \pi)$ is non-zero. Moreover $r_b(\pi', \pi) \leq 2$. Thus for any permutation $\pi$ we have  %Thus, for a permutation $\pi' \in N_{ba}$, $\sum_{\pi \in \Pi}\x_{\pi'}r_a(\pi', \pi) \geq 1$ and $\sum_{\pi \in \Pi}\x_{\pi'}r_b(\pi', \pi) \leq 2$.
\begin{equation}\label{muhum}
\sum_{\pi \in \Pi}r_b(\pi', \pi) \leq 2.
\end{equation}
As a result by Equation \eqref{eqqqq1} and Inequality \eqref{muhum}, we have
\begin{equation}
\rbab = \sum_{\pi'\in N_{ba}}  \sum_{\pi\in \Pi} \xb_{\pi'} r_b(\pi', \pi) \leq \sum_{\pi' \in N_{ba}}2\x_{\pi'} \label{j1}
\end{equation}
\begin{equation}
\rabb = \sum_{\pi'\in N_{ab}}  \sum_{\pi\in \Pi} \xb_{\pi'} r_b(\pi', \pi) \leq \sum_{\pi' \in N_{ab}}2\x_{\pi'} \label{j2}
\end{equation}
On the other hand for every rank $i$ if there is a permutation $\pi' \in N_{ba}$ with $\pi'(a)=i$, then there is exactly one permutation $\pi \in \Pi$ such that $i \in [\pi(b), \pi(a)]$ and hence $r_a(\pi', \pi) \geq 1$. Thus,
\begin{equation}
\rbaa = \sum_{\pi'\in N_{ba}}  \sum_{\pi\in \Pi} \xb_{\pi'} r_a(\pi', \pi) \geq \sum_{\pi' \in N_{ba}}\x_{\pi'} \label{j3}
\end{equation}
By Inequalities \eqref{j1} and \eqref{j3} we have
\begin{equation}
\rbab \leq 2\rbaa \label{j4}
\end{equation}

Now we can prove the lemma. By Equation \eqref{eqn_R},
\begin{align*}
\rabb+ \rbab &\geq \frac{p_a}{p_b}(\raba+ \rbaa) &\text{Since $\raba \geq 0$}\\
\rabb+ \rbab &\geq \frac{p_a}{p_b}\rbaa & \quad\quad\quad\quad\quad\quad\text{By Inequality \eqref{j4}} \\
\rabb + 2\rbaa &\geq \frac{p_a}{p_b}\rbaa \\
\rabb &\geq (\frac{p_a}{p_b} - 2)\rbaa &\text{By Inequality \eqref{j2}} \\
\sum_{\pi' \in N_{ab}}2\x_{\pi'} & \geq (\frac{p_a}{p_b} - 2)\rbaa &\text{By Inequality \eqref{j3}} \\
\sum_{\pi' \in N_{ab}}2\x_{\pi'} & \geq (\frac{p_a}{p_b} - 2)\sum_{\pi' \in N_{ba}}\x_{\pi'}
\end{align*}
Since $\sum_{\pi' \in N_{ab}}\x_{\pi'} = Pr_{\pi \sim \xb}[\pi(a)< \pi(b)]$ and $\sum_{\pi' \in N_{ba}}\x_{\pi'}=Pr_{\pi \sim \xb}[\pi(b) < \pi(a)]$,
$$Pr_{\pi \sim \xb}[\pi(a)< \pi(b)] \geq (\frac{p_a}{2p_b} - 1)Pr_{\pi \sim \xb}[\pi(b) < \pi(a)].$$

%% file: alg.tex
\begin{algorithm}
\textbf{input:} $\x, a, b$
\begin{algorithmic}[1]
%\label{alg1}
\State $S^* \gets N_{ba}$.
\State $\Pi \gets \emptyset$.
\While {$S^* \neq \emptyset$}
	\State Let $\pi \in S^*$ be the permutation with the rightmost $\pi(a)$ among all permutations in $S^*$. 
	\State $\Pi \gets \Pi \cup \{\pi\}$.
	\State $S^* \gets \{ \pi' \in S^*|[\pi'(b), \pi'(a)] \cap [\pi(b), \pi(a)] = \emptyset \}$.
\EndWhile
\State \Return $\Pi$ 

\end{algorithmic}
\caption{} 
%\caption{\hporacle} 
%\label{alg:hp_sep}
\label{alg1}
\end{algorithm}

%Let $N$ be the set of all permutations such as $\pi$ with $\xb_{\pi} > 0$, and let $S$ be the set of all permutations in $N$ with $\pi(b) < %\pi(a)$. We select a subset of $S$ via the following algorithm:

%\begin{enumerate}
%\item Let $S^* = S$, and $\hat{S} = \emptyset$.
%\item Let $\pi \in S^*$ be the permutation with the rightmost $\pi(a)$ among all permutations in $S^*$. 
%\item Let $\hat{S} = \hat{S} \cup \pi$.
%\item Let $S^* = S^* - \{ \pi' \in S^*|(\pi'(b), \pi'(a)) \cap (\pi(b), \pi(a)) \neq \emptyset \}$.
%\item if $S^* = \emptyset$ return $\hat{S}$. Otherwise go to line $2$. 
%\end{enumerate}

%\input{alg.tex}

%% file: plot.tex
\usetikzlibrary{shapes.geometric}
\tikzstyle{bl2} = [draw=black!90,thick,dashed]

\begin{figure}%[!ht]
\centering
%\begin{subfigure}[b]
%\includegraphics[width=\imw, height=\imh]{facebook_LT_2} 
%\caption{\facebook}
\tikzset{every mark/.append style={scale=1}}
\pgfplotsset{
  axis x line=bottom,
  axis y line=left,
  every outer y axis line/.append style={|-|},
  every outer x axis line/.append style={|-|},
  title style={font=\bf\Large},
%  tick label style={font=\small},
%  label style={font=\normalsize},
  every axis label/.append style={font=\large},
    legend style={font=\normalsize },%,line width=1pt,mark size=2pt},
%  }
}
\pgfplotstableread{alpha.dat}{\gpLT}
%\fbox{
\begin{tikzpicture}[scale=0.9]
%\draw (-.7in,-.45in) rectangle (4.4in,2.5in);
\begin{axis}[
    xlabel=$k$, 
    ylabel=Lower Bound on $\alpha_k$, 
%    title=Google+,
%    legend entries={Bound},
    xmin=2,
    ymin=0.45,
    xmax=100,
    legend pos=south east,
    legend cell align=left,
    height=2.5in,
    width=3in,
    scaled ticks=true,
]
\addplot [mark=.,very thick,blue,smooth] table [x={n}, y={bound}]{\gpLT};
\end{axis}
\draw[bl2] (6,0) -- (6,4.78);
\draw[bl2] (6,4.78) -- (0,4.78) node[left] {\small{$\POCBoundnew$}};

\draw[bl2] (0.5,4.17) -- (0.5,0);%  node[below] {$7$};

\draw (0.5,-0.06) node[below]{\small{$10$}};

\draw[bl2] (0.5,4.17) -- (0,4.17) node[left] {\small{$\POCBound$}};
\end{tikzpicture}%}
%\end{subfigure}

\caption{{\bf Lower bound on the solution of MP \ref{mp:2}.}  While we formally prove $\alpha_{10} \geq \POCBound$, this figure shows lower bounds on $\alpha_k$ for $2 \leq k \leq 100$ found by a computer program.
Note that, by Lemma \ref{lm:mp} the $\POC$ of the ranking duel with linear valuation function is at least $\alpha_k$ for all $n \geq k$.}
\label{fig_expr}
\end{figure}
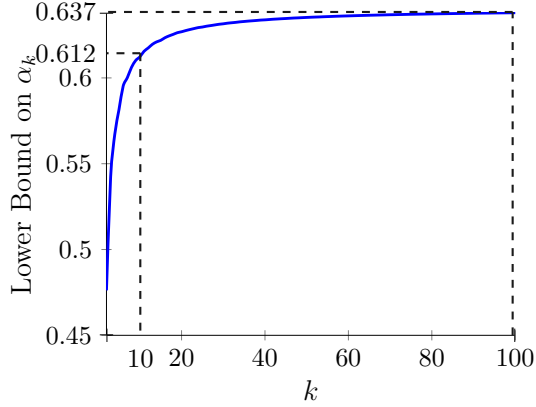

%% file: minimization.tex
In this section we consider the cost minimization version of the ranking duel with linear cost function and prove a constant upper bound for the $\POCCost$ of this game, using some structural results of the minimax strategies provided in Section \ref{sec:maxim}.
%The cost minimization version is defined cost by function $c$ rather than valuation function $v$, and the user is interested in the search result with a lower cost. Hence, the winner of the competition is the search that provides the search result with the least cost. The utility of player A for request $\omega$ is defined as follows:
%Note that, the only difference between the cost minimization and welfare maximization of the ranking duel game is that once a webpage is searched the winner of the cost minimization game is specified as follows:
%\begin{equation*}
%\uu{A}{\omega}{\px_A}{\px_B}=
%\begin{cases}
%+1 & \text{if } c_{\omega}(\px_A) < c_{\omega}(\px_B) \\
%0 & \text{if } c_{\omega}(\px_A) = c_{\omega}(\px_B)\\
%-1 & \text{if } c_{\omega}(\px_A) > c_{\omega}(\px_B)\\
%\end{cases}
%\end{equation*}
%Let $SC(\x) = E_{s \sim \x}[\sum_{\omega}p_{\omega}c_{\omega}(s)]$ be the social cost of strategy $\x$. Now, we define the $\POCCost$ of the ranking duel game as follows:
%\begin{equation}
%\POCCost = \frac{SC(\xminmax)}{SC(\OPTCost)}
%\end{equation}
%Where $\OPTCost$ is the strategy with least social cost and $\xminmax$ is the minimax strategy with the least social cost. In the following theorem we show that $\POCCost \leq 3$.
\begin{theorem}\label{costmin}
For any instance of the cost minimization ranking duel with linear cost function $\POCCost \leq 3$.
\end{theorem}
\begin{proof}
The minimum social cost is $\sum_{b=1}^n{p_bb}$, obtained by sorting the webpages according to their probabilities. Similar to Proposition \ref{prp:1}, for every strategy $\x$ the social cost is $1+\sum_{a=1}^n \sum_{b=i+1}^n p_a \nab + p_b \nba$. Thus the price of competition for a minimax strategy $\x$ is
\begin{equation}
\label{eq:1}
\frac{1+\sum_{a=1}^n \sum_{b=i+1}^n p_a \nab + p_b \nba}{\sum_{b=1}^np_bb}.
\end{equation}
For every pair of webpages $a$ and $b$ with $p_a \geq p_b$ we claim that in a minimax strategy $\x$, $p_a \nab + p_b \nba \leq 3p_b$.
Note that the polytope of the minimax strategies does not differ in the welfare maximization and cost minimization model. In other words, strategy $\x$ is a minimax strategy for the cost minimization ranking duel \textit{iff} it is a minimax strategy in the welfare maximization ranking duel. Thus Lemma \ref{lem_na_nb} also holds in the cost minimization ranking duel. By Lemma \ref{lem_na_nb}, $ (\frac{p_a}{2p_b}-1)\nab \leq  \nba$. %(It is the minimization version in which it is better to put high probable webpages in the rightmost parts). Thus
Simplifying the formula, we have $\frac{p_a}{2p_b}\nab \leq  \nba + \nab$. Since $\nba+\nab=1$, we have $$\frac{p_a}{p_b}\nab \leq 2.$$
Adding $\nba$ to both hand sides, and then multiplying both hand sides by $p_b$, we have
$$p_a\nab + p_b\nba \leq 2p_b + p_b\nba.$$
As $\nba \leq 1$, we can conclude
\begin{equation}\label{thebesteq}
p_a\nab + p_b\nba \leq 3p_b.
\end{equation}
Now by writing Inequality \eqref{thebesteq} property for all pairs of webpages $a$ and $b$ we have
\begin{align*}
\sum_{a=1}^n \sum_{b=i+1}^n p_a \nab + p_b \nba \leq 3 \sum_{b=1}^n p_b(b-1)
= 3(\sum_{b=1}^n{p_bb} -\sum_{b=1}^n{p_b}).
\end{align*}
Since $\sum_{b=1}^n{p_b}=1$, by simplifying formula we have
%$$\sum_{a=1}^n \sum_{b=i+1}^n p_a \nab + p_b \nba \leq 3 \sum_{b=1}^n p_b(b-1) = 3(\sum_{b=1}^n{p_bb} -\sum_{b=1}^n{p_b}).$$
$$\frac{3+\sum_{a=1}^n \sum_{b=i+1}^n p_a \nab + p_b \nba}{\sum_{b=1}^np_bb} \leq 3.$$
Thus by Equation \eqref{eq:1},
$\POCCost \leq 3.$
%$$\frac{1+\sum_{a=1}^n \sum_{b=i+1}^n p_a \nab + p_b \nba}{1+\sum_{b=1}p_bb} \leq 3.$$
\end{proof}

%% file: general_framework.tex
 \section{General framework}\label{Genfrm}
 In this section we present a general framework for analyzing the price of competition in dueling games. Proving lower bounds for the price of competition in dueling games highly depends on the valuation functions and it becomes more challenging when the valuation functions are complex.  However, the behavior of minimax strategies only depends on the comparison of the valuation functions rather than actual values. We leverage this fact to provide Theorem \ref{0-1 principle} which enables us to prove bounds for the price of competition without concerning the complexities of the valuation functions. We refer to this theorem as the \textit{0-1 principle}.
 
 Let $(\Omega,p,S,v)$ be a dueling game and $\alpha$ be a non-negative real number. We define the trigger function $\triggerv{\omega}{\alpha}(\px)$ for a pure strategy $\px$ in the following way:%\vspace{-0.5cm}
 \begin{equation*}
 	\triggerv{\omega}{\alpha}(\px)=
 	\begin{cases}
 		1 & \text{if } v_{\omega}(\px) \geq \alpha \\
 		0 & \text{if } v_{\omega}(\px) < \alpha\\
 	\end{cases}
 \end{equation*}
 Moreover, we define the pseudo-welfare function $\triggersw{\px}{\alpha}$ as the summation of the values of the trigger functions when a player is playing strategy $\px$ with respect to $\alpha$, $\triggersw{\px}{\alpha} = \sum_{\omega \in \Omega} p_{\omega}\triggerv{\omega}{\alpha}(\px)$. Furthermore, the pseudo-welfare function for a mixed strategy $\x$ is defined as 
 $\triggersw{\x}{\alpha} = E_{\px \sim \x}[\triggersw{\px}{\alpha}]$. Let $\triggerpoc{\alpha}$ be the pseudo-welfare of the minimax strategy with the least social welfare over $\SW(\OPT)$ which can be formulated by
 %\begin{equation*}
 $\triggerpoc{\alpha} = \frac{\triggersw{\xminmax}{\alpha}}{\triggersw{\OPT}{\alpha}}$,
 %\end{equation*}
 where $\xminmax$ is the minimax strategy with the least social welfare and $\OPT$ is the strategy with highest social welfare. Note that optimal and minimax strategies are determined regardless of the pseudo-welfare function. For simplicity, we consider $\triggerpoc{\alpha}  = 1$ when $\triggersw{\OPT}{\alpha} = 0$. In the following we show that the $\POC$ of every dueling game is bounded by $\min_{\alpha \geq 0}\{\triggerpoc{\alpha}\}$.
 \begin{theorem}\label{0-1 principle}
 	(0-1 principle) For every dueling game we have 
 	$\POC \geq \min_{\alpha \geq 0}\{\triggerpoc{\alpha}\}$.
 \end{theorem}

 \begin{proof}
 	Since $\triggerv{\omega}{\alpha}(s) = 1$ if and only if $v_{\omega}(s) \geq \alpha$, and $\triggerv{\omega}{\alpha}(s) = 0$ otherwise, we can formulate $v_{\omega}(s)$ as:
 	\begin{equation}\label{ghop1}
 		v_{\omega}(s) = \int_0^\infty \! \triggerv{\omega}{\alpha}(s) \, \mathrm{d}\alpha.
 	\end{equation}
 	Therefore, for every strategy $\x$ we have:
 	\begin{equation*}
 		\SW(\x) = E_{s \sim \x}[\sum_{\omega \in \Omega} p_{\omega}v_\omega(\px)] \hspace{1.3cm}
	\end{equation*}
	\begin{equation*}
	\hspace{5.3cm}	= E_{s \sim \x}[\int_0^\infty \! \sum_{\omega \in \Omega} p_{\omega}\triggerv{\omega}{\alpha}(\px) \, \mathrm{d}\alpha] \hspace{2cm} \text{By Equation (\ref{ghop1})}
	\end{equation*}
	\begin{equation*}
	\hspace{5.3cm}	 = \int_0^\infty \! E_{s \sim \x}[\triggersw{s}{\alpha}] \, \mathrm{d}\alpha
	\hspace{1cm} \text{By the definition of $\triggersw{s}{\alpha}$}
	\end{equation*}
	\begin{equation*}
	\hspace{0.0cm}	 =
 		% \int_0^\infty \!E_{s \sim \x}[\triggersw{s}{\alpha}] \, \mathrm{d}\alpha = 
 		\int_0^\infty \! \triggersw{\x}{\alpha} \, \mathrm{d}\alpha.\hspace{0.5cm}
 	\end{equation*}
 	Let $\beta = \min_{\alpha \geq 0}\{\triggerpoc{\alpha}\}$. Thus, for every minimax strategy $\xminmax$ and $\alpha \geq 0$ we have
 	$\beta \times \triggersw{\OPT}{\alpha} \leq \triggersw{\xminmax}{\alpha}$. Hence,
 	\begin{equation}
 		\beta \times \SW(\OPT)  = \beta \times \int_0^\infty \! \triggersw{\OPT}{\alpha} \, \mathrm{d}\alpha  
 	\end{equation}
 	\begin{equation*}
 		\hspace{2.9cm}\leq \int_0^\infty \! \triggersw{\xminmax}{\alpha} \, \mathrm{d}\alpha = \SW(\xminmax),
 	\end{equation*}
 	which implies $\POC \geq \beta$.
 \end{proof}
 
 %Note that, for every $\alpha$, one could design a new valuation function in such a way that $|\triggerv{\omega}{\alpha}(\x) - v_\omega(\x)| \leq \epsilon$ while the optimal and minimax strategies remain the same. Therefore, we have the lowest $\POC$ when the range of the valuation function is $[0,\epsilon] \cup [1,1+\epsilon]$. We leverage this fact to provide upper bounds for the $\POC$ of binary search and compression duel games in Subsections \ref{fu1} and \ref{fu2}.
 
 In the following subsections we show how we can apply the 0-1 principle to dueling games in order to present lower bounds for the $\POC$. In Subsection \ref{generalduel} we show that the $\POC$ of the ranking duel is at least $\frac{1}{4}$ regardless of the valuation function. %Next, we study the binary search and compression duel games with general valuation functions and show that the $\POC$ of these games cannot be bounded by any constant number.
 Note that, for every $\alpha$, one could design a valuation function in such a way that $|\triggerv{\omega}{\alpha}(\x) - v_\omega(\x)| \leq \epsilon$ while the optimal and minimax strategies remain the same. Therefore, we have the lowest $\POC$ when the range of the valuation function is $[0,\epsilon] \cup [1,1+\epsilon]$. We use this fact to provide upper bounds for the $\POC$ of binary search and compression duels in Subsections \ref{fu1} and \ref{fu2}.
 
 \subsection{Ranking duel with general valuation function}\label{generalduel}
 Recall that in the ranking duel each position of the permutation has a valuation $\positionutility{i}$, each pure strategy of the players is a permutation of webpages $\pi = \langle \pin(1),\pin(2),\pin(3),\ldots,\pin(n) \rangle$, and $\Omega = \{1,2,\ldots,n\}$ is the set of elements of uncertainty. For a webpage $\omega \in \Omega$, 
 $v_\omega(\pi) = \positionutility{\rank{\pi}{\omega}}$, where $\pi(\omega)$ is the rank of $\omega$ in $\pi$. In the following, we use the 0-1 principle to show that the $\POC$ of the ranking duel with an arbitrary valuation function is at least $\frac{1}{4}$.
 \begin{theorem}\label{f14}
 	The $\POC$ of the ranking duel is at least $\frac{1}{4}$.
 \end{theorem}
 
 \begin{proof}
 \input{proof1.tex}
 \end{proof}
 
 \subsection{Compression duel with general valuation function}\label{fu1}
 Compression duel is a dueling games which was introduced by Immorlica et al. \cite{duel11}. In this game each pure strategy of the players is a binary search tree with leaf set $\Omega = \{1,2,\ldots,n\}$. Therefore, $S$ is the set of all binary search trees that have elements of $\Omega$ as leaves and $p$ is a distribution of probabilities over $\Omega$. Once a request $\omega \in \Omega$ is drawn from the probability distribution $p$, $v_w(s)$ of each player is determined with  $f(\depth{s}{\omega})$ where $f: \mathbb{N} \rightarrow \Rnonnegative$ is a non-increasing function and $\depth{s}{\omega}$ is the depth of the leaf $\omega$ in $s$. In the following, we show that the $\POC$ of this game can be arbitrarily close to zero.
 \begin{theorem}\label{compressiontheorem}
 	For every $\epsilon > 0$, there exists an instance of the compression duel, with $\POC \leq \epsilon$.
 \end{theorem}
 \input{proof2.tex}

 \subsection{Binary search duel with general valuation function}\label{fu2}
 In this subsection we study the binary search duel and show that the $\POC$ of this game can be $\Omega(\frac{1}{n})$. In this game $\Omega = \{1,2,\ldots,n\}$ and each pure strategy of the players is a binary tree such that its in-order traversal visits the elements of $\Omega$ in the sorted order. Moreover, $v_\omega(s)$ is determined by $f(\depth{s}{\omega})$ where $\depth{s}{\omega}$ denotes the depth of element $\omega$ in the binary search tree corresponding to $s$ and $f:\mathbb{N} \rightarrow \Rnonnegative$ is a decreasing function.
 \begin{theorem}\label{binarytheorem}
 	For every $\beta > 0$ there is an instance of the binary search duel with $|\Omega|= \theta(\frac{1}{\beta})$ and $\POC \leq \beta$.
 \end{theorem}
 \input{proof3.tex}

%% file: proof1.tex
To prove this lower bound we first show that $\triggerpoc{\alpha}$ of this game is at least $\frac{1}{4}$ for all $\alpha \geq 0$ and then we apply the 0-1 principle and conclude that $\POC \geq \frac{1}{4}$.
Let $\alpha \geq 0$ be a real number and $k$ be the number of indices $i$ of the permutation such that $\positionutility{i} \geq \alpha$. Therefore, for each pure strategy $\pi$, $\triggerv{\omega_i}{\alpha} = 1$ for exactly $k$ elements $\omega_i$ and $\triggerv{\omega_i}{\alpha} = 0$ for the other $n-k$  elements. Since $p_1 \geq p_2 \geq \ldots \geq p_n$, 
\begin{equation}\label{aaab}
\triggersw{\OPT}{\alpha} \leq \sum_{i=1}^k p_i.
\end{equation}
Also for every pure strategy $\pi$,
$\triggersw{\pi}{\alpha} \leq \sum_{i=1}^k p_i$. 
%Since the pseudo-social welfare of every mixed strategy is the weighted average of the pseudo-social welfare of pure strategies, we have the same bound for the pseudo-social welfare of mixed strategies and thus we have
%\begin{equation}\label{aaab}
%\triggersw{\OPT}{\alpha} \leq \sum_{i=1}^k p_i.
%\end{equation}
Now let $\xminmax$ be the minimax strategy with the least social welfare. For each webpage $i$, let $q_i$ be the probability that $\xminmax$ puts webpage $i$ in a position with valuation at least $\alpha$. Therefore we can formulate $\triggersw{\xminmax}{\alpha}$ as
\begin{equation}\label{aab}
\triggersw{\xminmax}{\alpha} = \sum_{i=1}^n p_i q_i.
\end{equation}
Consider the mixed strategy $\x'_i$ which draws a random permutation $\pi$ from $\xminmax$ and plays permutation $\transformed{\pi}$ as follows
\begin{itemize}
\item If $f(\rank{\pi}{i}) \geq \alpha$ then return $\pi$.
\item Otherwise, Let $w_1,w_2,\ldots,w_k$ be the set of webpages such that $f(\rank{\pi}{w_i})  \geq \alpha$. Choose one of the webpages in $\{w_1,w_2,\ldots,w_k\}$ uniformly at random and swap the position of that webpage with the position of webpage $i$ and return the new permutation.
\end{itemize}
Note that strategy $\x'_i$ plays in such a way that $f(\rank{\pi}{w_i}) \geq \alpha$, and thus $\triggerv{\omega}{\alpha}(\x'_i) = 1$. Since $\xminmax$ is a minimax strategy, the payoff $u^A(\xminmax,\x'_i)$ should be greater than or equal to 0. 
%We use this fact to show $p_i(1-q_i)^2 \leq \sum_{\omega = 1}^n \frac{2p_\omega q_\omega}{k} $ which helps us to prove $\triggerpoc{\alpha} \geq \frac{1}{4}$. First, we show that 
This observation, follows from the fact that in an NE of the game, the payoff of both players is equal to 0, hence, every minimax strategy should guarantee a payoff of at least 0. In the following, we bound the payoff of the game between $\x^*$ and $\x_i'$ in terms of $p_i$'s and $q_i$'s as follows 
\begin{equation}\label{1men}
\payoff{i}(\xminmax,\x'_i) \leq -(1-q_i)^2,
\end{equation}
\begin{equation*}
\forall \omega \neq i \hspace{0.5cm}\payoff{\omega}(\xminmax,\x'_i) \leq 2 q_\omega.
\end{equation*}
 Given that the payoff of $\x^*$ is at least 0, we conclude the following inequality
%Since $u(\xminmax,\x'_i) \geq 0$ we conclude that 
$$p_i(1-q_i)^2 \leq \sum_{\omega = 1}^n \frac{2p_\omega q_\omega}{k}.$$
%First, we show that
%\begin{equation}\label{1men}
%\payoff{i}(\xminmax,\x'_i) \leq -(1-q_i)^2
%\end{equation}
By definition, $\payoff{i}(\xminmax,\x'_i)$ is specified by the following formula
\begin{equation*}
Pr_{\begin{subarray}{l}\pi_A \sim \xminmax \\ \pi_B \sim \xminmax\end{subarray}}[v_i(\pi_A) > v_i(\transformed{\pi_B})] -
Pr_{\begin{subarray}{l}\pi_A \sim \xminmax \\ \pi_B \sim \xminmax\end{subarray}}[v_i(\pi_A) < v_i(\transformed{\pi_B})].
\end{equation*}
In order to prove Inequality \eqref{1men}, we consider the following two cases
\begin{enumerate}
\item $\triggerv{i}{\alpha}(\pi_A) = \triggerv{i}{\alpha}(\pi_B) = 0$. Since $\triggerv{i}{\alpha}(\transformed{\pi_B})$ is always equal to 1, we have $Pr[v_i(\pi_A) > v_i(\transformed{\pi_B})]-Pr[v_i(\pi_A) < v_i(\transformed{\pi_B})] = -1$.
\item Otherwise, since $\pi_A$ and $\pi_B$ are both drawn from the same strategy, we know that the expected value of $Pr[v_i(\pi_A) > v_i(\pi_B)]-Pr[v_i(\pi_A) < v_i(\pi_B)]$ in this case is 0. Furthermore, $v_i(\transformed{\pi_B}) \geq v_i(\pi_B)$, thus $Pr[v_i(\pi_A) > v_i(\transformed{\pi_B})]-Pr[v_i(\pi_A) < v_i(\transformed{\pi_B})] \leq 0$.
\end{enumerate}
Since for $\pi \sim \xminmax$, $\triggerv{i}{\alpha}(\pi) = 1$ with probability $q_i$, the first case happens with probability $(1-q_i)^2$ and the second case happens with probability $1-(1-q_i)^2$, we have
\begin{equation*}
Pr_{\begin{subarray}{l}\pi_A \sim \xminmax \\ \pi_B \sim \xminmax\end{subarray}}[v_i(\pi_A) > v_i(\transformed{\pi_B})] -
Pr_{\begin{subarray}{l}\pi_A \sim \xminmax \\ \pi_B \sim \xminmax\end{subarray}}[v_i(\pi_A) > v_i(\transformed{\pi_B})] \leq -(1-q_i)^2,
\end{equation*}
which implies Inequality \eqref{1men}. Next, we show that $u^A_\omega(\xminmax,\x'_i) \leq 2q_\omega$ for all $\omega \neq i$. Again, by definition, we have
\begin{equation*}
u^A_\omega(\xminmax,\x'_i) = Pr_{\begin{subarray}{l}\pi_A \sim \xminmax \\ \pi_B \sim \xminmax\end{subarray}}[v_\omega(\pi_A) > v_\omega(\transformed{\pi_B})] -
Pr_{\begin{subarray}{l}\pi_A \sim \xminmax \\ \pi_B \sim \xminmax\end{subarray}}[v_\omega(\pi_A) < v_\omega(\transformed{\pi_B})].
\end{equation*}
Note that, $v_\omega(\transformed{\pi_B}) = v_\omega(\pi_B)$ with probability at least $1-\frac{q_\omega}{k}$ for every $\omega \neq i$ and replacing $\pi_B$ by $\transformed{\pi_B}$ will increase the payoff of the strategy $\xminmax$ by at most 2 (changing the payoff from -1 to +1). Therefore
\begin{align*}
&\bigg( Pr_{\begin{subarray}{l}\pi_A \sim \xminmax \\ \pi_B \sim \xminmax\end{subarray}}[v_\omega(\pi_A) > v_\omega(\transformed{\pi_B})] &-&
Pr_{\begin{subarray}{l}\pi_A \sim \xminmax \\ \pi_B \sim \xminmax\end{subarray}}[v_\omega(\pi_A) < v_\omega(\transformed{\pi_B})]&&\bigg) -\\
&\bigg(Pr_{\begin{subarray}{l}\pi_A \sim \xminmax \\ \pi_B \sim \xminmax\end{subarray}}[v_\omega(\pi_A) > v_\omega(\pi_B)] &-&
Pr_{\begin{subarray}{l}\pi_A \sim \xminmax \\ \pi_B \sim \xminmax\end{subarray}}[v_\omega(\pi_A) < v_\omega(\pi_B)]&&\bigg)  \leq \frac{2q_\omega}{k}.
\end{align*}
Given both $\pi_A$ and $\pi_B$ are drawn from the same strategy, 
\begin{equation}
\bigg(Pr_{\begin{subarray}{l}\pi_A \sim \xminmax \\ \pi_B \sim \xminmax\end{subarray}}[v_\omega(\pi_A) > v_\omega(\pi_B)] -
Pr_{\begin{subarray}{l}\pi_A \sim \xminmax \\ \pi_B \sim \xminmax\end{subarray}}[v_\omega(\pi_A) < v_\omega(\pi_B)]\bigg) = 0,
\end{equation}
Therefore for every $\omega \neq i$, 
\begin{equation}\label{secondeq}
u^A_\omega(\xminmax,\x'_i) = 
\bigg(Pr_{\begin{subarray}{l}\pi_A \sim \xminmax \\ \pi_B \sim \xminmax\end{subarray}}[v_\omega(\pi_A) > v_\omega(\transformed{\pi_B})] -
Pr_{\begin{subarray}{l}\pi_A \sim \xminmax \\ \pi_B \sim \xminmax\end{subarray}}[v_\omega(\pi_A) < v_\omega(\transformed{\pi_B})]\bigg)   \leq 2 q_\omega.
\end{equation}
By applying Inequalities (\ref{1men}) and (\ref{secondeq}), we have
\begin{align*}
u^A(\xminmax,\x'_i) &= \sum_{\omega=1}^n p_\omega u^A_\omega(\xminmax,\x'_i)\\
&\leq -p_i(1-q_i)^2+\sum_{\omega \neq i} \frac{2p_\omega q_\omega}{k}\\
&\leq -p_i(1-q_i)^2+\sum_{\omega = 1}^n \frac{2p_\omega q_\omega}{k}.
\end{align*}
Since $\xminmax$ is a minimax strategy, $u^A(\xminmax,\x'_i) \geq 0$ and thus
\begin{equation}\label{mohem}
p_i(1-q_i)^2 \leq \sum_{\omega = 1}^n \frac{2p_\omega q_\omega}{k}.
\end{equation}
By summing Inequality (\ref{mohem}) for all $1 \leq i \leq k$ we obtain
$\sum_{i=1}^k p_i(1-q_i)^2 \leq \sum_{\omega = 1}^n 2p_\omega q_\omega.$
Therefore,
\begin{equation*}
\sum_{i=1}^k p_i (1+q_i^2-2q_i) \leq \sum_{\omega = 1}^n 2p_\omega q_\omega
\end{equation*}
By moving $-2q_i$ to the right hand side of the inequality we have
\begin{equation*}
\sum_{i=1}^k p_i (1+q_i^2) \leq \sum_{\omega = 1}^n 2p_\omega q_\omega + \sum_{i=1}^k 2p_iq_i
\end{equation*}
Since $q_i^2$ is non-negative, we can remove it from the left hand side. Thus,
\begin{equation*}
\sum_{i=1}^k p_i \leq \sum_{\omega = 1}^n 2p_\omega q_\omega + \sum_{i=1}^k 2p_iq_i
\leq  \sum_{i=1}^n 4p_iq_i
\end{equation*}
Recall that, by Inequality (\ref{aaab}) we have $\triggersw{\OPT}{\alpha} \leq \sum_{i=1}^k p_i$ and by Equation (\ref{aab}) we have $\triggersw{\xminmax}{\alpha} = \sum_{i=1}^n p_i q_i$. Hence,
\begin{equation*}
\triggerpoc{\alpha} = \frac{\triggersw{\xminmax}{\alpha}}{\triggersw{\OPT}{\alpha}}  \geq \frac{1}{4}
\end{equation*}
By applying the 0-1 principle, we conclude that $\POC$ of the ranking duel game with general valuation function is at least $\frac{1}{4}$.

%% file: proof2.tex
\begin{proof}
Consider a compression duel game where $\Omega = \{1,2,3,4\}$ and $P = \langle \frac{1}{4},\frac{1}{4},\frac{1}{4},\frac{1}{4} \rangle$. Let the valuation function $f$ be as follows:
\begin{equation*}
f(d) =
\begin{cases}
1 & \text{if } d \leq 2 \\
\frac{\epsilon}{16} & \text{if }  d = 3\\
0 & \text{if } d> 3. \\
\end{cases}
\end{equation*}
\input{Fig1.tex}
\input{Fig2.tex}
We show that $\POC$ of this game is no more than $\epsilon$. Since the binary tree in Figure \ref{compopt}, has one leaf with depth $2$, one leaf with depth $3$, and two leaves with depth $4$, the social welfare of its corresponding strategy is $\frac{1}{4}\times 1 +  \frac{1}{4} \times \frac{\epsilon}{16} = \frac{16+\epsilon}{64}$. Hence, we have
\begin{equation}\label{comsw}
\SW(\OPT) \geq \frac{16+\epsilon}{64}
\end{equation}
Next, we prove that the pure strategy $\xminmax$ corresponding to the binary tree in Figure \ref{compminmax} is a minimax strategy of the game. To do so, we show that $u_A(\xminmax,\y) \geq 0$ for every pure strategy $\y$. Since the binary tree corresponding to $\y$ should have exactly 4 leaves, it has at most one leaf with depth lower than $3$. Therefore one of the two condition holds for $\y$.
\begin{enumerate}
\item The binary tree corresponding to $\y$ has exactly one leaf of depth $2$. In this case at least two of the webpages have depth more than $3$, and thus $u^A(\xminmax,\y) \geq 0$. 
\item All the webpages have depth of 3 or more in $\y$. In this case $u^A_\omega(\xminmax,\y) \geq 0$ for all $\omega = \{1,2,3,4\}$. Hence $u^A(\xminmax,\y) \geq 0$.
\end{enumerate}
Therefore $\xminmax$ is a minimax strategy of the game. The social welfare of $\xminmax$ is equal to $(\frac{1}{4}+\frac{1}{4}+\frac{1}{4}+\frac{1}{4}) \times \frac{\epsilon}{16} = \frac{\epsilon}{16}$. By Inequality \eqref{comsw} we have:
\begin{equation*}
\POC \leq \frac{\frac{\epsilon}{16}}{\SW(\OPT)} \leq \frac{\frac{\epsilon}{16}}{\frac{16+\epsilon}{64}}\leq \frac{4 \epsilon}{16+\epsilon} \leq \epsilon
\end{equation*}
\end{proof}

%% file: Fig1.tex
\tikzstyle{H-node}=[rectangle,draw=black,fill=white!30,inner sep=1.3mm]
\tikzstyle{B-node}=[circle,draw=blue,fill=blue!20,inner sep=2.5mm]
\tikzstyle{G-node}=[circle,draw=green,fill=green!30,inner sep=1.3mm]
\tikzstyle{R-node}=[rectangle,draw=red,fill=red!20,inner sep=2.6mm]
\tikzstyle{W-node}=[rectangle,draw=white,fill=white!30,inner sep=0.2mm]
\tikzstyle{test-node}=[circle,draw=black,fill=black,inner sep=.2mm]

\tikzstyle{bl0} = [draw=black, thick, dashed]   
\tikzstyle{b9} = [draw=black, thick]   
\tikzstyle{bl1} = [->, draw=black]   
\tikzstyle{bl2} = [draw=black!70,thick]   
\tikzstyle{bl3} = [draw=black,thick, dotted]   

\tikzstyle{br0} = [draw=brown, dashed]   
\tikzstyle{br1} = [->, draw=brown]   
\tikzstyle{br2} = [->, draw=brown,thick]   

\tikzstyle{red0} = [draw=red, thick, dashed]   
\tikzstyle{red1} = [draw=red]   
\tikzstyle{red2} = [draw=red,thick]   

\tikzstyle{gr0} = [draw=green, thick, dashed]   
\tikzstyle{gr1} = [draw=green]   
\tikzstyle{gr2} = [draw=green,thick]   
\tikzstyle{gr4} = [draw=green,semithick,rounded corners]   

\begin{figure}
\begin{center}
\begin{tikzpicture}[scale=0.5][domain=0:8]%,dash pattern=on 2pt off 3pt on 4pt off 4pt]
\draw (0,5) node[B-node,label=center:$$,label=above:] (b_1) {};
\draw (3,2) node[R-node,label=center:$4$,label=above:] (b_2) {};
\draw (-3,2) node[B-node,label=center:$$,label=above:] (b_3) {};
\draw (-6,-1) node[B-node,label=center:$$,label=above:] (b_4) {};
\draw (0,-1) node[R-node,label=center:$3$,label=above:] (b_5) {};
\draw (-9,-4) node[R-node,label=center:$1$,label=above:] (b_6) {};
\draw (-3,-4) node[R-node,label=center:$2$,label=above:] (b_7) {};

\node[below] at (-12,5.7) {$\rightarrow$ depth 1};
\node[below] at (-12,2.7) {$\rightarrow$ depth 2};
\node[below] at (-12,-0.3) {$\rightarrow$ depth 3};
\node[below] at (-12,-3.3) {$\rightarrow$ depth 4};

%\draw (1,5) node[B-node,label=center:$A$,label=above:$\pi_A$\text{$=$}$0.0938$] (b_1) {};
%\draw (5,5) node[B-node,label=center:$B$,label=above:$\pi_B$\text{$=$}$0.0938$] (b_2) {};
%\draw (3,0) node[R-node,label=center:$1$, label=below:$p_1$\text{$=$}$\frac{1}{2}$] (h_1) {};

\draw[b9] (b_1)  to node [label=left:] {} (b_2) ;
\draw[b9] (b_1)  to node [label=left:] {} (b_3) ;
\draw[b9] (b_3)  to node [label=left:] {} (b_4) ;
\draw[b9] (b_3)  to node [label=left:] {} (b_5) ;
\draw[b9] (b_4)  to node [label=left:] {} (b_6) ;
\draw[b9] (b_4)  to node [label=left:] {} (b_7) ;

%\draw[b9] (b_2)  to node [label=right:$q_{B1}$] {} (h_1) ;
%\draw (3, 8) node[W-node,label=center:First Scenario] (s1) {};

%\draw (9,5) node[B-node,label=center:$A$,label=above:$\pi_A$\text{$=$}$0.0938$] (bb_1) {};
%\draw (13,5) node[B-node,label=center:$B$,label=above:$\pi_B$\text{$=$}$0.0938$] (bb_2) {};
%\draw (9,0) node[R-node,label=center:$1$, label=below:$p_1$\text{$=$}$\frac{1}{2}$] (hh_1) {};
%\draw (13,0) node[R-node,label=center:$2$, label=below:$p_2$\text{$=$}$\frac{1}{2}$] (hh_2) {};

%\draw[b9] (bb_1) to node [label=left:$q_{A1}$] {} (hh_1) ;
%\draw[b9] (bb_1) -- (hh_2) ;
%\draw[b9] (bb_2) -- (hh_1) ;
%\draw[b9] (bb_2) to node [label=right:$q_{B2}$] {} (hh_2) ;
%\draw (11, 8) node[W-node,label=center:Second Scenario] (s2) {};

%\draw (9.8, 3) node[W-node,label=center:$q_{A2}$] (lab2) {};
%\draw (12.2, 3) node[W-node,label=center:$q_{B1}$] (lab3) {};

%\draw (17,5) node[B-node,label=center:$A$,label=above:$\pi_A$\text{$=$}$0.124$] (bbb_1) {};
%\draw (21,5) node[B-node,label=center:$B$,label=above:$\pi_B$\text{$=$}$0.064$] (bbb_2) {};
%\draw (17,0) node[R-node,label=center:$1$, label=below:$p_1$\text{$=$}$0.64$] (hhh_1) {};
%\draw (21,0) node[R-node,label=center:$2$, label=below:$p_2$\text{$=$}$0.48$] (hhh_2) {};

%\draw[b9] (bbb_1)  to node [label=left:$q_{A1}$] {} (hhh_1) ;
%\draw[b9] (bbb_1)  to node [label=left:$q_{A2}$] {} (hhh_2) ;
%\draw[b9] (bbb_2)  to node [label=right:$q_{B2}$] {} (hhh_2) ;
%\draw (19, 8) node[W-node,label=center:Third Scenario] (s3) {};
\end{tikzpicture}
\end{center}
\caption{Depth of the webpages is 4,4,3,2, respectively.}
\label{compopt}
\end{figure}
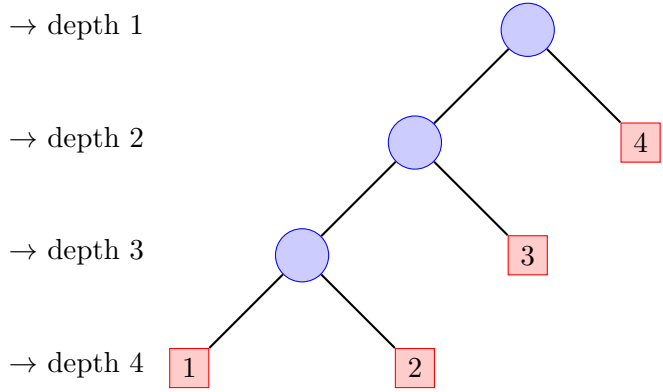

%% file: Fig2.tex
\tikzstyle{H-node}=[rectangle,draw=black,fill=white!30,inner sep=1.3mm]
\tikzstyle{B-node}=[circle,draw=blue,fill=blue!20,inner sep=2.5mm]
\tikzstyle{G-node}=[circle,draw=green,fill=green!30,inner sep=1.3mm]
\tikzstyle{R-node}=[rectangle,draw=red,fill=red!20,inner sep=2.6mm]
\tikzstyle{W-node}=[rectangle,draw=white,fill=white!30,inner sep=0.2mm]
\tikzstyle{test-node}=[circle,draw=black,fill=black,inner sep=.2mm]

\tikzstyle{bl0} = [draw=black, thick, dashed]   
\tikzstyle{b9} = [draw=black, thick]   
\tikzstyle{bl1} = [->, draw=black]   
\tikzstyle{bl2} = [draw=black!70,thick]   
\tikzstyle{bl3} = [draw=black,thick, dotted]   

\tikzstyle{br0} = [draw=brown, dashed]   
\tikzstyle{br1} = [->, draw=brown]   
\tikzstyle{br2} = [->, draw=brown,thick]   

\tikzstyle{red0} = [draw=red, thick, dashed]   
\tikzstyle{red1} = [draw=red]   
\tikzstyle{red2} = [draw=red,thick]   

\tikzstyle{gr0} = [draw=green, thick, dashed]   
\tikzstyle{gr1} = [draw=green]   
\tikzstyle{gr2} = [draw=green,thick]   
\tikzstyle{gr4} = [draw=green,semithick,rounded corners]   

\begin{figure}
\begin{center}
\begin{tikzpicture}[scale=0.5][domain=0:8]%,dash pattern=on 2pt off 3pt on 4pt off 4pt]
\draw (0,5) node[B-node,label=center:$$,label=above:] (b_1) {};
\draw (3,2) node[B-node,label=center:$$,label=above:] (b_3) {};
\draw (-3,2) node[B-node,label=center:$$,label=above:] (b_2) {};
\draw (-4.5,-1) node[R-node,label=center:$1$,label=above:] (b_4) {};
\draw (-1.5,-1) node[R-node,label=center:$2$,label=above:] (b_5) {};
\draw (1.5,-1) node[R-node,label=center:$3$,label=above:] (b_6) {};
\draw (4.5,-1) node[R-node,label=center:$4$,label=above:] (b_7) {};

%\draw (1,5) node[B-node,label=center:$A$,label=above:$\pi_A$\text{$=$}$0.0938$] (b_1) {};
%\draw (5,5) node[B-node,label=center:$B$,label=above:$\pi_B$\text{$=$}$0.0938$] (b_2) {};
%\draw (3,0) node[R-node,label=center:$1$, label=below:$p_1$\text{$=$}$\frac{1}{2}$] (h_1) {};

\draw[b9] (b_1)  to node [label=left:] {} (b_2) ;
\draw[b9] (b_1)  to node [label=left:] {} (b_3) ;
\draw[b9] (b_2)  to node [label=left:] {} (b_4) ;
\draw[b9] (b_2)  to node [label=left:] {} (b_5) ;
\draw[b9] (b_3)  to node [label=left:] {} (b_6) ;
\draw[b9] (b_3)  to node [label=left:] {} (b_7) ;

\node[below] at (-10.5,5.7) {$\rightarrow$ depth 1};
\node[below] at (-10.5,2.7) {$\rightarrow$ depth 2};
\node[below] at (-10.5,-0.3) {$\rightarrow$ depth 3};

%\draw[b9] (b_2)  to node [label=right:$q_{B1}$] {} (h_1) ;
%\draw (3, 8) node[W-node,label=center:First Scenario] (s1) {};

%\draw (9,5) node[B-node,label=center:$A$,label=above:$\pi_A$\text{$=$}$0.0938$] (bb_1) {};
%\draw (13,5) node[B-node,label=center:$B$,label=above:$\pi_B$\text{$=$}$0.0938$] (bb_2) {};
%\draw (9,0) node[R-node,label=center:$1$, label=below:$p_1$\text{$=$}$\frac{1}{2}$] (hh_1) {};
%\draw (13,0) node[R-node,label=center:$2$, label=below:$p_2$\text{$=$}$\frac{1}{2}$] (hh_2) {};

%\draw[b9] (bb_1) to node [label=left:$q_{A1}$] {} (hh_1) ;
%\draw[b9] (bb_1) -- (hh_2) ;
%\draw[b9] (bb_2) -- (hh_1) ;
%\draw[b9] (bb_2) to node [label=right:$q_{B2}$] {} (hh_2) ;
%\draw (11, 8) node[W-node,label=center:Second Scenario] (s2) {};

%\draw (9.8, 3) node[W-node,label=center:$q_{A2}$] (lab2) {};
%\draw (12.2, 3) node[W-node,label=center:$q_{B1}$] (lab3) {};

%\draw (17,5) node[B-node,label=center:$A$,label=above:$\pi_A$\text{$=$}$0.124$] (bbb_1) {};
%\draw (21,5) node[B-node,label=center:$B$,label=above:$\pi_B$\text{$=$}$0.064$] (bbb_2) {};
%\draw (17,0) node[R-node,label=center:$1$, label=below:$p_1$\text{$=$}$0.64$] (hhh_1) {};
%\draw (21,0) node[R-node,label=center:$2$, label=below:$p_2$\text{$=$}$0.48$] (hhh_2) {};

%\draw[b9] (bbb_1)  to node [label=left:$q_{A1}$] {} (hhh_1) ;
%\draw[b9] (bbb_1)  to node [label=left:$q_{A2}$] {} (hhh_2) ;
%\draw[b9] (bbb_2)  to node [label=right:$q_{B2}$] {} (hhh_2) ;
%\draw (19, 8) node[W-node,label=center:Third Scenario] (s3) {};
\end{tikzpicture}
\end{center}
\caption{Depth of all the webpages is 3}
\label{compminmax}
\end{figure}
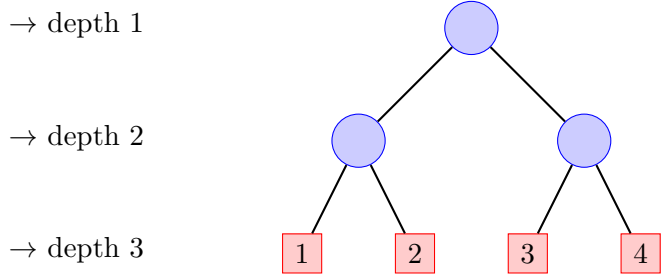

%% file: proof3.tex
\begin{proof}
Let $k = \lceil \lg \frac{1}{\beta} \rceil + 2$ and $\Omega = \{1,2,\ldots,n\}$ be the set of webpages where $n = 3 \times 2^k$. Moreover, let $P = \langle \frac{1}{5},\frac{4}{5(n-1)},\frac{4}{5(n-1)},\ldots,\frac{4}{5(n-1)}\rangle$ and the valuation function $f$ be as follows:
\begin{equation*}
f(d) =
\begin{cases}
1 & \text{if } d = 1 \\
\epsilon & \text{if }  2 \leq d \leq k+2\\
0 & \text{if } d > k+2. \\
\end{cases}
\end{equation*}
Where $\epsilon$ is a very small positive number.
Note that, since $\lceil \lg \frac{1}{\beta} \rceil \geq 1$, $n$ is always greater than or equal to $24$. Moreover, The valuation function for the root of the tree is equal to $1$. Thus, by putting the first webpage as the root of the tree we can achieve the social welfare of at least $\frac{1}{5}$. Therefore, we have:
\begin{equation}
\SW(\OPT) \geq \frac{1}{5}
\end{equation}
\input{Fig3.tex}

Next, we show that there exists a minimax strategy of this game with social welfare equal to $\frac{4}{5(n-1)} + \frac{5(n-1)-4}{5(n-1)}\epsilon$. Consider the pure strategy $\xminmax$ such that plays the binary tree of Figure \ref{fig:binary}. In this strategy, depth of the first webpage is always 2 and the depth all other webpages is at most $k+2$. Therefore, the social welfare of this strategy is equal to $\frac{4}{5(n-1)} + \frac{5(n-1)-4}{5(n-1)}\epsilon$. Furthermore, we show that $u^A(\xminmax,\y) \geq 0$ for all pure strategies $\y$ and conclude that $\xminmax$ is a minimax strategy. To do so, we divide the pure strategies into two categories:
\begin{enumerate}
\item The pure strategies that have the first webpage as the root of the tree. Since the in-order traversal of the webpages in these binary trees should be from first webpage to the last webpage, all the other webpages are in the right subtree of the root. Therefore in these strategies depth of at least $2^k$ webpages is higher than $k+2$. Thus, $u^A_\omega(\xminmax,\y) = 1$ for at least $2^k$ webpages and $u^A_1(\xminmax,\y) = -1$. Since $n \geq 24$, we have $p_1 = \frac{1}{5} < 2^k\times\frac{4}{5(n-1)} = \frac{n}{3}\times\frac{4}{5(n-1)}$. Therefore, $u^A(\xminmax,\y) \geq 0$.
\item The pure strategies in which depth of the first webpage is more than $1$. In these strategies one webpage $\omega$ such that $p_\omega = \frac{4}{5(n-1)}$ is the root and the value of the valuation function for all other webpages is at most $\epsilon$. Since the value of the valuation function for all the webpages in $\xminmax$ is at least $\epsilon$ and its root has the same probability to be requested, we have $u^A(\xminmax,\y) \geq 0$.
\end{enumerate}
Since $u^A(\xminmax,\y) \geq 0$ for all pure strategies $\y$, $\xminmax$ is a minimax strategy. This implies that the following inequality holds for the $\POC$ of this game:
\begin{equation*}
\POC \leq \frac{\SW(\xminmax)}{\SW(\OPT)} \leq \frac{\frac{4}{5(n-1)} + \frac{5(n-1)-4}{5(n-1)}\epsilon}{\frac{1}{5}} \leq \frac{4}{n-1}+5\epsilon
\end{equation*}
Since $n \geq 24$ and $\epsilon$ is small enough, we have:
\begin{equation*}
\POC \leq \frac{4}{n-1}+5\epsilon \leq \frac{5}{n}
\end{equation*}
Recall that, $n = 3\times 2^{\lceil \lg \frac{1}{\beta} \rceil+2}$, therefore $\frac{5}{n} < \beta$ and thus $\POC \leq \beta$.
\end{proof}

%% file: Fig3.tex
\usetikzlibrary{shapes.geometric}

\tikzstyle{H-node}=[rectangle,draw=black,fill=white!30,inner sep=1.3mm]
\tikzstyle{B-node}=[circle,draw=blue,fill=blue!20,inner sep=5.5mm]
\tikzstyle{G-node}=[circle,draw=green,fill=green!30,inner sep=1.3mm]
\tikzstyle{R-node}=[rectangle,draw=red,fill=red!20,inner sep=6.5mm]
\tikzstyle{W-node}=[rectangle,draw=white,fill=white!30,inner sep=0.2mm]
\tikzstyle{test-node}=[circle,draw=black,fill=black,inner sep=.2mm]

\tikzstyle{bl0} = [draw=black, thick, dashed]   
\tikzstyle{b9} = [draw=black, thick]   
\tikzstyle{bl1} = [->, draw=black]   
\tikzstyle{bl2} = [draw=black!70,thick]   
\tikzstyle{bl3} = [draw=black,thick, dotted]   

\tikzstyle{br0} = [draw=brown, dashed]   
\tikzstyle{br1} = [->, draw=brown]   
\tikzstyle{br2} = [->, draw=brown,thick]   

\tikzstyle{red0} = [draw=red, thick, dashed]   
\tikzstyle{red1} = [draw=red]   
\tikzstyle{red2} = [draw=red,thick]   

\tikzstyle{gr0} = [draw=green, thick, dashed]   
\tikzstyle{gr1} = [draw=green]   
\tikzstyle{gr2} = [draw=green,thick]   
\tikzstyle{gr4} = [draw=green,semithick,rounded corners]   

\tikzset{
    buffer/.style={
        draw,
        shape border rotate=0,
        regular polygon,
        regular polygon sides=3,
        fill=white,
        node distance=1cm,
        minimum height=4em
    }
}

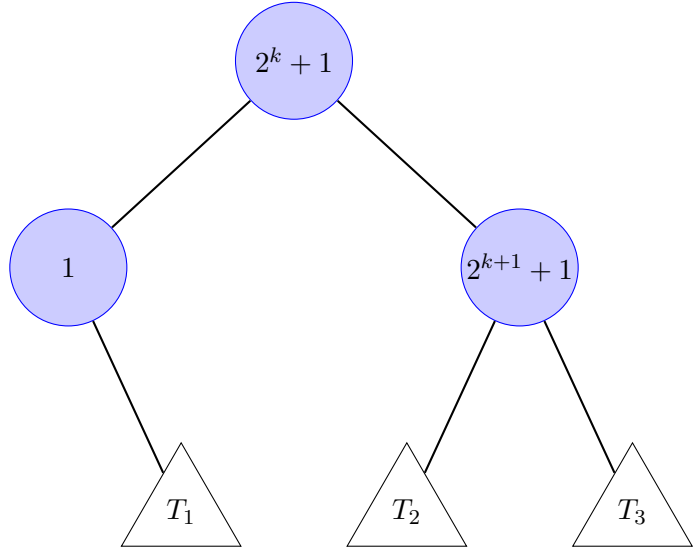
\begin{figure}
\begin{center}
\begin{tikzpicture}[scale=0.5][domain=0:8]%,dash pattern=on 2pt off 3pt on 4pt off 4pt]

\draw (0,10) node[B-node,label=center:$2^k+1$,label=above:] (b_1) {};
\draw (6,4.5) node[B-node,label=center:$2^{k+1}+1$,label=above:] (b_3) {};
\draw (-6,4.5) node[B-node,label=center:$1$,label=above:] (b_2) {};
%\draw (-4.5,-1) node[R-node,label=center:$1$,label=above:] (b_4) {};
%\draw (-1.5,-1) node[R-node,label=center:$2$,label=above:] (b_5) {};
%\draw (1.5,-1) node[R-node,label=center:$3$,label=above:] (b_6) {};
%\draw (4.5,-1) node[R-node,label=center:$4$,label=above:] (b_7) {};

%\node[buffer] at (-1.5,-2)  (b_5){$[2,2^k]$} ;
%\node[buffer] at (1.5,-2)  (b_6){$[2^k+2,2^{k+1}]$} ;
%\node[buffer] at (4.5,-2)  (b_7){$[2^{k+1}+1,3\times 2^k]$} ;

\node[buffer] at (-3,-2)  (b_5){$T_1$} ;
\node[buffer] at (3,-2)  (b_6){$T_2$} ;
\node[buffer] at (9,-2)  (b_7){$T_3$} ;

%\node[buffer] at (-3,-30)  (b_50){$---$};

%\draw (1,5) node[B-node,label=center:$A$,label=above:$\pi_A$\text{$=$}$0.0938$] (b_1) {};
%\draw (5,5) node[B-node,label=center:$B$,label=above:$\pi_B$\text{$=$}$0.0938$] (b_2) {};
%\draw (3,0) node[R-node,label=center:$1$, label=below:$p_1$\text{$=$}$\frac{1}{2}$] (h_1) {};
\draw[b9] (b_1)  to node [label=left:] {} (b_2) ;
\draw[b9] (b_1)  to node [label=left:] {} (b_3) ;
%\draw[b9] (b_2)  to node [label=left:] {} (b_4) ;
\draw[b9] (b_2)  to node [] {} (b_5) ;
\draw[b9] (b_3)  to node [label=left:] {} (b_6) ;
\draw[b9] (b_3)  to node [label=left:] {} (b_7) ;

%\draw[b9] (b_2)  to node [label=right:$q_{B1}$] {} (h_1) ;
%\draw (3, 8) node[W-node,label=center:First Scenario] (s1) {};

%\draw (9,5) node[B-node,label=center:$A$,label=above:$\pi_A$\text{$=$}$0.0938$] (bb_1) {};
%\draw (13,5) node[B-node,label=center:$B$,label=above:$\pi_B$\text{$=$}$0.0938$] (bb_2) {};
%\draw (9,0) node[R-node,label=center:$1$, label=below:$p_1$\text{$=$}$\frac{1}{2}$] (hh_1) {};
%\draw (13,0) node[R-node,label=center:$2$, label=below:$p_2$\text{$=$}$\frac{1}{2}$] (hh_2) {};

%\draw[b9] (bb_1) to node [label=left:$q_{A1}$] {} (hh_1) ;
%\draw[b9] (bb_1) -- (hh_2) ;
%\draw[b9] (bb_2) -- (hh_1) ;
%\draw[b9] (bb_2) to node [label=right:$q_{B2}$] {} (hh_2) ;
%\draw (11, 8) node[W-node,label=center:Second Scenario] (s2) {};

%\draw (9.8, 3) node[W-node,label=center:$q_{A2}$] (lab2) {};
%\draw (12.2, 3) node[W-node,label=center:$q_{B1}$] (lab3) {};

%\draw (17,5) node[B-node,label=center:$A$,label=above:$\pi_A$\text{$=$}$0.124$] (bbb_1) {};
%\draw (21,5) node[B-node,label=center:$B$,label=above:$\pi_B$\text{$=$}$0.064$] (bbb_2) {};
%\draw (17,0) node[R-node,label=center:$1$, label=below:$p_1$\text{$=$}$0.64$] (hhh_1) {};
%\draw (21,0) node[R-node,label=center:$2$, label=below:$p_2$\text{$=$}$0.48$] (hhh_2) {};

%\draw[b9] (bbb_1)  to node [label=left:$q_{A1}$] {} (hhh_1) ;
%\draw[b9] (bbb_1)  to node [label=left:$q_{A2}$] {} (hhh_2) ;
%\draw[b9] (bbb_2)  to node [label=right:$q_{B2}$] {} (hhh_2) ;
%\draw (19, 8) node[W-node,label=center:Third Scenario] (s3) {};
\end{tikzpicture}
\end{center}
\caption{All the subtrees are complete binary trees with height $k$. Subtree $T_1$ contains all the webpages from $2$ to $2^k$, subtree $T_2$ contains all the webpages from $2^{k}+2$ to $2^{k+1}$ and subtree $T_3$ contains all the webpages from $2^{k+1}+2$ to $3\times 2^k$}
\label{fig:binary}
\end{figure}

%% file: ack.tex
\section*{Acknowledgment}
We would like to gratefully thank Rakesh Vohra and Brendan Lucier for their insightful discussions.

%% file: lp_table.tex
\begin{figure}[!h]
\begin{center}
\begin{tabular}{|l|l|l|l|}
\hline
$\theta$ & $\beta$ & $\lambda$ & $\rho$\\
\hline
$\theta= 0.612275$ & $\beta_{1,2}= 1$ & $\lambda_{1,5}= 1$ &$\rho_{1,4}= 0.86974$\\
& $\beta_{1,3}= 1$ &$\lambda_{1,6}= 1$ & $\rho_{2,6}= 0.939106$\\
& $\beta_{1,4}= 0.13026$ & $\lambda_{1,7}= 1$ & $\rho_{3,7}= 0.332976$\\
& $\beta_{2,3}= 1$ & $\lambda_{1,8}= 1$ & $\rho_{3,8}= 1$\\
& $\beta_{2,4}= 1$ & $\lambda_{1,9}= 1$ & $\rho_{3,9}= 1$\\
& $\beta_{2,5}= 1$ & $\lambda_{1,10}= 1$ & $\rho_{4,7}= 0.0537254$\\
& $\beta_{2,6}= 0.0608937$ & $\lambda_{2,7}= 1$ & $\rho_{4,8}= 0.615102$\\
& $\beta_{3,4}= 1$ & $\lambda_{2,8}= 1$ & $\rho_{4,9}= 1$\\
& $\beta_{3,5}= 1$ & $\lambda_{2,9}= 1$ & $\rho_{4,10}= 0.274381$\\
& $\beta_{3,6}= 1$ & $\lambda_{2,10}= 1$ & $\rho_{5,9}= 0.422703$\\
& $\beta_{3,7}= 0.667024$ &$\lambda_{3,10}= 1$ & $\rho_{5,10}= 1$\\
& $\beta_{4,5}= 1$ & $\lambda_{4,10}= 0.725619$ & $\rho_{6,9}= 0.420799$\\
& $\beta_{4,6}= 1$ & &$\rho_{6,10}= 0.394388$\\
& $\beta_{4,7}= 0.946275$ & &\\
& $\beta_{4,8}= 0.384898$ & &\\
& $\beta_{5,6}= 1$ & &\\
& $\beta_{5,7}= 1$ & &\\
& $\beta_{5,8}= 1$ & &\\
& $\beta_{5,9}= 0.577297$ & &\\
& $\beta_{6,7}= 1$ & &\\
& $\beta_{6,8}= 1$ & &\\
& $\beta_{6,9}= 0.579201$ & &\\
& $\beta_{6,10}= 0.605612$ & &\\
& $\beta_{7,8}= 1$ & &\\
& $\beta_{7,9}= 1$ & &\\
& $\beta_{7,10}= 1$ & &\\
& $\beta_{8,9}= 1$ & &\\
& $\beta_{8,10}= 1$ & &\\
& $\beta_{9,10}= 1$ & &\\
\hline
\end{tabular}
\end{center}
\caption{\textbf{Feasible solution for LP \ref{mp:5}.} In this table we present the non-zero variables of LP \ref{mp:5},  which is the dual program of LP \ref{mp:4}. This feasible solution gives us a lower bound for the primal linear program.}
\label{tablemastak}
\end{figure}